\documentclass[11pt,a4paper]{article}
\usepackage[utf8]{inputenc}
\usepackage{amsmath,amssymb,amsthm}
\usepackage{booktabs}
\usepackage{graphicx}
\usepackage{bm}
\usepackage[expansion=false]{microtype}
\usepackage[margin=1in]{geometry}
\usepackage{hyperref}
\hypersetup{colorlinks=true,linkcolor=blue,citecolor=blue,urlcolor=blue}

\newtheorem{theorem}{Theorem}[section]
\newtheorem{proposition}[theorem]{Proposition}
\newtheorem{lemma}[theorem]{Lemma}
\newtheorem{corollary}[theorem]{Corollary}
\theoremstyle{definition}

\theoremstyle{remark}
\newtheorem{remark}[theorem]{Remark}

\DeclareMathOperator{\Tr}{Tr}
\DeclareMathOperator{\diag}{diag}
\DeclareMathOperator{\Rot}{Rot}

\newcommand{\I}{\mathbb{1}}
\newcommand{\openone}{\mathbb{1}}% REVTeX alias for Paper II's transferred text

\newcommand{\nEB}{n_{\mathrm{EB}}}
\newcommand{\floor}{n_{\mathrm{EB}}^{\mathrm{floor}}}
\newcommand{\nuc}[1]{\left\lVert #1 \right\rVert_{*}}
\newcommand{\opn}[1]{\left\lVert #1 \right\rVert_{2}}
\newcommand{\Frob}[1]{\left\lVert #1 \right\rVert_{F}}
\newcommand{\EB}{\mathrm{EB}}
\newcommand{\PPT}{\mathrm{PPT}}
\providecommand{\ket}[1]{|#1\rangle}
\providecommand{\bra}[1]{\langle#1|}

\title{The quantum lifetime of a future-referential feedback loop:\\
certified index, architecture floor, and thermal bonus}
\author{Eran Kopel\\[3pt]
  {\normalsize Tel Aviv University, Tel Aviv, Israel}\\[2pt]
  {\normalsize \href{mailto:erankopel@tauex.tau.ac.il}{erankopel@tauex.tau.ac.il}}\\[2pt]
  {\normalsize ORCID: \href{https://orcid.org/0000-0003-4657-8636}{0000-0003-4657-8636}}}
\date{\today}

\begin{document}
\maketitle

\begin{abstract}
A quantum feedback loop that returns information from a forward simulation to
an earlier internal time induces a completely positive trace-preserving map on
a message register; iterating that map eventually destroys its ability to
carry entanglement. The entanglement-breaking index $N=\nEB(\Phi)$, the round
at which this happens, is studied along two axes, interaction and bath. On the
interaction axis, strict contraction onto a full-rank fixed point forces $N$
finite in every dimension, via an explicit separable ball around the limiting
Choi state, and across $3997$ sampled channels index and stability gap are one
clock, $N\,[1-\rho(A)]\in[0.60,1.50]$. On the bath axis, a thermal ancilla at
polarisation $p=\tanh(\hbar\omega/2k_BT)$ makes every channel quantity an exact
quadratic in $p$, and one-sided certificates (no eigensolver) establish $N=3$
at the reference circuit at every temperature. The infinite-temperature index
is an architecture floor with closed form and exact measure ($46.427\%$ of circuits break entanglement in one round at infinite temperature). For exchange collisions (generalised amplitude damping) the index has a closed form
and is monotone in $p$; for the loop the floor is not a bound: certified
circuits dip below it in interior temperature windows, at a rate with Gaussian
cutoff in $\varepsilon=1-\opn{A_0}$ (marginal exponent $2.02\,[1.90,2.16]$,
$4.1$ million circuits), and shallow-floor valleys survive to
$\varepsilon=0.125$. On a certified coupling-temperature surface, temperature
moves the index by at most $3/2$, the interaction by $1213/3$. Applied to
quantum-cognition proposals, a per-bath polarisation table at $310$~K shows
nuclear spins maximally hot and intramolecular vibrations frozen, and a
four-input checklist replaces decoherence times by round budgets; Fisher's
nuclear-spin proposal specifies its bath and storage coherence, not its
processing interaction or rate.
\end{abstract}

\section{Introduction}\label{sec:intro}

A quantum system coupled to an environment and driven repeatedly (a feedback
loop, a relay, a repeated measurement) eventually loses the ability to carry
entanglement with anything outside itself. The natural question is not
\emph{whether} but \emph{after how many rounds}, and the natural answer is an
integer.

That integer is the \textbf{entanglement-breaking index} \cite{LG2015}
\begin{equation}\label{eq:neb}
  \nEB(\Phi) \;=\; \min\{\, n \ge 1 \;:\; \Phi^{n} \text{ is entanglement breaking} \,\},
\end{equation}
the smallest number of compositions after which $\Phi^n$ becomes a
measure-and-prepare channel and no entanglement with a reference survives
\cite{HSR2003}. The index is a strictly stronger object than the thresholds
usually reported. This point has recently been made forcefully: for uniform
dephasing and dephasing-plus-damping channels the Choi state is non-PPT at
\emph{every} finite decoherence rate, so any finite ``entanglement-breaking
threshold'' quoted for them is either the coherent-information threshold under
a wrong name, or an artefact of numerical tolerance \cite{NoThreshold}. The
critique is correct, and it is constructive: for a strictly contractive
channel the well-posed quantity is not a single-shot threshold but the
composition index. This paper takes that quantity seriously: computes it,
certifies it, and asks what it depends on.

Two literatures bear on the index and neither answers the question we pose.
The \textbf{structural} literature establishes when an index exists
\cite{RJP2018,HRF2020,CMW2019,Park2026}. The finiteness theorems in this line
take PPT as a hypothesis. The channel studied here is \textbf{NPT}, and
therefore outside every one of them; finiteness follows instead from strict
contraction toward a full-rank fixed point. The \textbf{physical} literature
parameterises channels by physical quantities but stops at single-shot
properties. The generalised amplitude damping channel, the qubit analogue of
the bosonic thermal channel, has known entanglement-breaking and
antidegradability regions in its damping and noise parameters \cite{KSW2020}.
Collision models supply the standard framework in which a system interacts
repeatedly with fresh thermal ancillas \cite{Ciccarello2022}. But an excitation-exchange collision with a thermal qubit gives a GADC; GADCs at fixed noise form a semigroup, and for a semigroup the index follows in closed form (Proposition~\ref{prop:gadc}). Nothing needs certifying, and nothing of interest happens: the index is monotone in the bath polarisation, so colder never hurts. What is absent is
the intersection: an index computed for a channel whose parameters come from a
physical bath, in a case where the family is \emph{not} a semigroup, and
computed rigorously rather than estimated.

We supply that case. The channel is the future-referential feedback loop of
the companion paper \cite{PaperI}: externally an ordinary causally ordered
circuit, internally future-referential, in which a record extracted from a
forward simulation is fed back to an earlier internal time of the simulated
dynamics. Contracting the process tensor with a leakage instrument and a
controller induces a CPTP map $\Phi$ on a message register $M$, and the
self-consistent histories of the protocol are the fixed points
$\rho^{\ast}=\Phi(\rho^{\ast})$. Five operational criteria were used in
\cite{PaperI} to classify them (stability, informativeness, feedability,
coherence, and preservation of quantum correlations), and an explicit
parameter square was certified on which the four nontrivial ones hold
simultaneously. All five are properties of the fixed point; none says how the
loop behaves \emph{on the way there}. The index is the sixth criterion, the
first that is a property of the \emph{trajectory} and the first that is an
integer rather than a yes/no; because $\EB$ is a mapping cone
\cite{Stormer86}, $\Phi^{n}\in\EB$ for every $n\ge\nEB$, so it is a genuine
threshold. This paper computes that threshold along two axes, the interaction
and the bath, and certifies every value it quotes.

On the interaction axis, at zero temperature, the results are of three kinds.
Section~\ref{sec:contractive} proves that strict contraction onto a full-rank
fixed point bounds $\nEB$ explicitly, in every dimension
(Theorem~\ref{thm:generald}), together with a closed form for a qubit register
in terms of the Bloch data (Theorem~\ref{thm:qubit}) whose overshoot on a
natural family is determined exactly (Corollary~\ref{cor:tight}).
Section~\ref{sec:certified} evaluates the index for the partial-swap family of
\cite{PaperI} and certifies $N=3$ uniformly on a four-dimensional parameter
box, in ball arithmetic and without trusting an eigenvalue solver.
Section~\ref{sec:landscape} maps $N$ over a two-dimensional slice and finds
that it is controlled, above and below, by the stability gap alone: the rate
at which the loop forgets its input and the number of rounds it stays quantum
are one clock, not two.

On the bath axis, the loop is already a collision model
(\S\ref{sec:model}): one collision per feedback round, with the ancilla at
zero temperature. Replacing the zero-temperature ancilla by a thermal one at
polarisation $p=\tanh(\hbar\omega/2k_BT)$ therefore buys a physical
temperature axis at no structural cost, and the feedback structure (a
controlled rotation and a partial swap) breaks the semigroup property that
would otherwise trivialise the problem. Three things follow. (i) The index
splits: at infinite bath temperature the channel, and hence the index, is a
functional of the interaction alone, an \emph{architecture floor}
(Theorem~\ref{thm:floor}), and the infinite-temperature channel collapses to a
dephasing map composed with a rotation and a uniform contraction, giving the
floor in closed form (Theorem~\ref{thm:normalform}) and its distribution as an
exact measure. The natural conjecture that the bath term above the floor is non-negative, a theorem for plain exchange collisions (\S\ref{sec:baseline}), is \textbf{false} here: certified counterexamples exist (\S\ref{sec:valley}) and their rate is measured. (ii) The two terms are wildly asymmetric: across a certified coupling-temperature surface the ratio of the index at $T=0$ to its value at $T=\infty$ never exceeds $21/14 = 3/2$ exactly, while the coupling strength
moves it by a factor of $1213/3$. (iii) The physical reading inverts: at $310$~K every intramolecular vibrational and electronic degree of freedom sits at $p\approx1$, the coldest possible bath by this measure, and it is the spin degrees of freedom that are maximally hot. Section~\ref{sec:physical} uses this, with the closed form of the exchange-collision baseline, to restate a twenty-five-year debate over decoherence times in the brain as a question about round budgets, and lists the four inputs a proposal must specify to receive one.

Every numerical claim here is a statement about the sign of a real quantity or
the exact value of an integer, and is verified rather than estimated.
Negativity is established by exhibiting a witness vector whose Rayleigh
quotient is certified negative; positivity by certifying all four leading
principal minors. Both primitives are one-sided: a failed certificate is
inconclusive, never wrong. No eigenvalue solver enters at any stage. At zero
temperature the certificate executes $33$ checks in $256$-bit ball arithmetic,
with a second-stack audit at $60$ decimal digits. On the temperature axis the
quadratic dependence on $p$ lets a single interval propagate through a
degree-two polynomial, so adaptive bisection can \emph{cover} a parameter
range rather than sample it: $\nEB = 3$ for every $p\in[0,1]$ at the reference
circuit ($223$ boxes), and for an uncountable family of circuits at every
temperature ($12{,}523$ boxes, none unresolved). To our knowledge these are
the first entanglement-breaking indices certified over a continuum rather than
at points.

A by-product bears on the one question this paper raises and does not settle.
Theorem~\ref{thm:qubit} bounds the index in terms of $\opn{A}$, and
\S\ref{sec:spectral} asks whether the spectral radius $\rho(A)$ may be
substituted, open because $\lVert A^n\rVert \neq \rho(A)^n$ for non-normal
$A$. The weak-coupling limit of the thermal family discriminates sharply: as
the coupling vanishes, the $\rho(A)$ form's prediction converges monotonically
to the certified exact value ($73.95 \to 98.99\%$) while the $\opn{A}$ form's
diverges from it ($99.88 \to 85.17\%$). The conjectured form is asymptotically
saturated in a regime where the proved form is not.

The model is one qubit interacting repeatedly with thermal ancillas under a
specified unitary, and the results are statements about that model. No claim
is made that any biological system realises this loop, that the index bounds
any biological process, or that the surviving rounds have functional
significance. What the construction offers is a certified integer where the surrounding literature has long had order-of-magnitude estimates. Section~\ref{sec:physical} applies the index to proposals of quantum effects in the brain as a unit conversion; it takes no side in that debate (\S\ref{sec:claims}).

\section{The feedback loop as a thermal collision model}\label{sec:setup}\label{sec:model}

\subsection{The feedback channel}

We recall only what is needed; see~\cite{PaperI} for the construction. A
process tensor $\Upsilon$ couples a simulated world $W$, a returned message
register $M$, a future-event register $F$, a leakage probe $L$, and the
remaining simulator environment $E$. One feedback round induces
\begin{equation}\label{eq:channel}
  \Phi(\rho_{M})\;=\;\sum_{l}\mathcal{C}_{l}\bigl[\mathcal{M}_{l}
    \bigl(S_{\Upsilon}(\rho_{W}\otimes\rho_{M})\bigr)\bigr],
\end{equation}
with $\mathcal{M}_{l}$ the leakage instrument and $\mathcal{C}_{l}$ the
controller. Throughout, $M$ is a qubit unless stated otherwise. In the Bloch
parametrisation $\rho=(\openone+\bm{v}\cdot\bm{\sigma})/2$ the channel is the
affine map
\begin{equation}\label{eq:bloch}
  \bm{v}\;\longmapsto\;A\bm{v}+\bm{c},
\end{equation}
with $A_{ij}=\tfrac12\Tr[\sigma_i\Phi(\sigma_j)]$ and
$c_i=\Tr[\sigma_i\Phi(\openone/2)]$.
When $\|A\|_{2}<1$ the fixed point is unique, with Bloch vector
$\bm{v}^{\ast}=(\openone-A)^{-1}\bm{c}$. Here and below $\|A\|_{2}$ denotes the
largest singular value of the $3\times3$ real matrix $A$.

Two descriptions of $\Phi$ appear below and it is worth stating once that they
agree. Equation~\eqref{eq:channel} exhibits $\Phi$ as a contraction of the
process tensor against an instrument and a controller; being CPTP, it also
admits a Stinespring dilation, and for the family studied here that dilation is
the single round unitary $U$ on $MFL$ with $F$ and $L$ initialised in
$\ket{00}$ and the world register absorbed into $U$, so that
\begin{equation}\label{eq:dilation}
  \Phi(\cdot)=\Tr_{FL}\bigl[U(\cdot\otimes\ket{00}\bra{00})U^{\dagger}\bigr].
\end{equation}
this is Eqs.~(4)--(9) of \cite[Sec.~IV]{PaperI}.
Equation~\eqref{eq:dilation} is the form evaluated by the certificate and used
throughout Appendix~\ref{app:cert}; Eq.~\eqref{eq:channel} is retained because
it is the form in which the operational criteria are stated.

The family of interest is the four-parameter partial-swap extension of
\cite[Sec.~VI]{PaperI}, with round unitary
$U(\theta,\kappa,\varphi,\beta)=R_{\beta}U_{f}(\varphi)U_{w}(\kappa)U_{W}(\theta)$
acting on $MFL$. Its certified reference point is
\begin{align}\label{eq:refpoint}
  (\theta_{0},\varphi_{0})/\pi &= (0.16345853,\,0.20061939),\notag\\
  (\kappa_{0},\beta_{0})/\pi &= (0.43230980,\,0.23903823),
\end{align}
about which \cite[Prop.~4]{PaperI} certifies a square of half-width
$0.0013\pi$ in $\theta$ and $\varphi$, at fixed $\kappa_{0}$ and $\beta_{0}$,
carrying simultaneously a contraction margin $1-\|A\|_{2}\ge0.237$, coherence
$C_{\ell_{1}}\ge0.416$, information $I(F{:}L)\ge0.172$ bits, and an NPT margin
$-\lambda_{\min}[J^{T_{R}}]\ge0.188$.

We will certify the index not on that square but on the full four-dimensional
box of the same half-width,
\begin{equation}\label{eq:box}
  \mathcal{B}\;=\;\bigl\{p:\|p-p_{0}\|_{\infty}\le s\bigr\},\qquad
  s=0.0013\pi,
\end{equation}
with $p=(\theta,\kappa,\varphi,\beta)$, which contains the square of
Proposition~4 as its $(\theta,\varphi)$ slice.

\subsection{The loop is a collision model}

A fresh ancilla in $\ket{00}$ is an ancilla at zero temperature: the
$\ket{00}$ of the construction above is the $p=1$ end of the thermal family
introduced now.

The channel is defined by a fixed three-qubit unitary $U$ acting on a message register $M$, a future register $F$ and a latent register $L$, with $F$ and $L$ prepared afresh each round in $\ket{00}$ and traced out, as in \eqref{eq:dilation}. This is precisely the structure of a \textbf{quantum collision model} \cite{Ciccarello2022}: a system interacting repeatedly with fresh, identically prepared ancillas, one collision per round. The identification is not an analogy. It is the same object, and it supplies the physical parameter the construction was missing.

Explicitly, with $R_y(a) = \cos\tfrac a2\,\I - i\sin\tfrac a2\,Y$,
\begin{align}
  U &\;=\; \big(R_y(\beta)_M\otimes\I_{FL}\big)\;U_f(\varphi)\;U_w(\kappa)\;U_W(a_0,a_1), \label{eq:U}\\
  U_W &\;=\; |0\rangle\!\langle0|_M\otimes R_y(a_0)_F\otimes\I_L \;+\; |1\rangle\!\langle1|_M\otimes R_y(a_1)_F\otimes\I_L, \nonumber\\
  U_w &\;=\; \exp\!\big(-i\tfrac{\kappa}{2}\,\I_M\otimes Z_F\otimes Y_L\big), \qquad
  U_f \;=\; \cos\varphi\,\I - i\sin\varphi\,\mathrm{SWAP}_{ML}. \nonumber
\end{align}
The construction of \cite{PaperI} is the \emph{supplementary} case $a_0 = \pi-2\theta$, $a_1 = 2\theta$; we keep the branch angles independent where it costs nothing, because Theorem~\ref{thm:normalform} and Lemma~\ref{lem:structuralzero} are both sharper that way.

A fresh ancilla in $|00\rangle$ is an ancilla at zero temperature. Replacing it with a thermal state promotes the entire apparatus to a function of temperature:
\begin{equation}\label{eq:thermalloop}
  \Phi_p(\rho) \;=\; \Tr_{FL}\big[\,U\,(\rho \otimes \tau_S(p)^{\otimes 2})\,U^{\dagger}\,\big],
  \qquad \tau_S(p) = \tfrac12\big(\I + p\,S\big),
\end{equation}
for a traceless Hermitian involution $S$ ($S^2 = \I$) fixing the bath's polarisation axis. For a two-level ancilla of splitting $\hbar\omega$ in equilibrium at temperature $T$,
\begin{equation}\label{eq:p}
  p \;=\; \tanh\!\Big(\frac{\hbar\omega}{2k_BT}\Big),
\end{equation}
so $p=1$ is $T=0$ and $p=0$ is $T=\infty$. Writing the qubit channel in Bloch affine form $\Phi_p : r \mapsto A(p)\,r + c(p)$, every quantity of the series (stability margin, coherence, fixed point \cite{PaperI}, and the index) becomes a function of the bath temperature at no structural cost.

\begin{remark}[why the thermal bath, not the radiation field]
At physiological temperature the dominant decohering interactions in condensed phase are collisional, Coulombic and vibrational; blackbody photon absorption and emission on a molecular dipole is slower by more than twenty orders of magnitude once the long-wavelength suppression $(\Delta x/\lambda)^2$ is included \cite{Tegmark2000,Schlosshauer2019}. The bath here is therefore a generic thermal ancilla, not a radiation mode. The infrared enters this paper at exactly one point, in \S\ref{sec:physical}.
\end{remark}

\subsection{Exact form of the thermal family}

\begin{proposition}\label{prop:quadratic}
$A(p)$ and $c(p)$ are quadratic matrix polynomials in $p$.
\end{proposition}

\begin{proof}
$\Phi_p$ depends on $p$ only through the ancilla state, and depends on it linearly. Since
\[
  \tau_S(p)^{\otimes 2} \;=\; \tfrac14\Big(\I\otimes\I \;+\; p\,(S\otimes\I + \I\otimes S) \;+\; p^{2}\, S\otimes S\Big),
\]
the map $p \mapsto \Phi_p$ is a quadratic polynomial with operator coefficients, and so are its Bloch data:
\[
  A(p) = A_0 + p\,A_1 + p^2 A_2, \qquad c(p) = c_0 + p\,c_1 + p^2 c_2. \qedhere
\]
\end{proof}

Two consequences. First, the temperature dependence is \emph{exactly} representable: no expansion or truncation is required. Second, certification extends over the temperature axis using the machinery already built for a single point: one interval in $p$ propagates through a polynomial of degree two. This is what makes the covering results of \S\ref{sec:numerics} affordable.

\section{The entanglement-breaking index}
\label{sec:index}

A completely positive map is entanglement breaking iff its Choi matrix is
separable~\cite{HSR2003}. We use the normalised Choi matrix
\begin{equation}\label{eq:choi}
  J(\Phi)=\frac1d\sum_{i,j=1}^{d}\ket{i}\bra{j}\otimes\Phi(\ket{i}\bra{j}),
  \qquad \Tr J(\Phi)=1,
\end{equation}
and write $J^{T_{R}}$ for its partial transpose on the reference factor.

For a qubit register the classification problem collapses. By the
Peres--Horodecki criterion~\cite{Peres96,Horodecki96} a state of two qubits is
separable if and only if it is PPT, so $\EB=\PPT$ in $M_{2}$ and
\begin{equation}\label{eq:indexformula}
  \nEB(\Phi)\;=\;\min\bigl\{\,n\ge1:\ \lambda_{\min}\bigl[J(\Phi^{n})^{T_{R}}\bigr]\ge0\,\bigr\}.
\end{equation}
The index is thus \emph{exactly computable}: no separability oracle, and no
entanglement witness, is required. Every step of
Sec.~\ref{sec:certified} rests on Eq.~\eqref{eq:indexformula}.

This is the same collapse that makes the base case of Park's induction
immediate (``every PPT map on $M_{2}$ is EB and hence
$\nEB(\Phi)=1$''~\cite{Park2026}), and it cuts both ways. It is what makes the
index of our loop computable to certified precision; it is also what makes the
qubit case unable to host the bound-entangled regime, since a PPT point of a
qubit channel is EB with index one. We return to this in Sec.~\ref{sec:outlook}.

\section{The contractive route to eventual entanglement breaking}
\label{sec:contractive}

The feedback channels of Sec.~\ref{sec:setup} are not PPT, so
\cite[Thm.~1.1]{Park2026} does not apply to them, yet their index is finite. The
reason is that they contract strictly onto a fixed point of full rank. This
section makes that mechanism quantitative, and shows it is not special to a
qubit register.

\subsection{Notation}

For a channel $\Phi:M_d\to M_d$ we use the normalised Choi matrix
\begin{equation}
  J(\Phi) \;=\; \frac1d\sum_{i,j=1}^{d}\ket{i}\bra{j}\otimes\Phi(\ket{i}\bra{j}),
  \qquad \operatorname{Tr}J(\Phi)=1 ,
\end{equation}
and write $\|\cdot\|_F$, $\|\cdot\|_\infty$, $\|\cdot\|_1$ for the Frobenius,
spectral and trace norms. For the $3\times3$ real Bloch matrix we follow the
usual convention and write $\|A\|_2\equiv\|A\|_\infty$ for its largest singular
value; this is the operator norm, \emph{not} the Schatten-$2$ norm. The Hilbert--Schmidt norm of a superoperator is
$\|\Lambda\|_{\mathrm{HS}}^2=\sum_{\alpha}\|\Lambda(E_\alpha)\|_F^2$ for any
orthonormal basis $\{E_\alpha\}$ of $M_d$. Following \cite{LG2015,Park2026},
$n_{\mathrm{EB}}(\Phi)=\inf\{n\ge1:\Phi^n\in\mathrm{EB}\}$; since $\mathrm{EB}$
is a mapping cone, $\Phi^{n}\in\mathrm{EB}$ for all $n\ge n_{\mathrm{EB}}$, so
the index is a genuine threshold.

Throughout, $\Phi$ has a unique fixed state $\rho^\ast$ and we set
\begin{equation}
  \eta \;:=\; \sup\bigl\{\|\Phi(X)\|_F \;:\; \operatorname{Tr}X=0,\ \|X\|_F=1\bigr\},
\end{equation}
the Frobenius contraction coefficient on the traceless subspace. For a qubit
channel with Bloch pair $(A,c)$ one has $\eta=\|A\|_2$, and we write
$a:=\|A\|_2$, $v^\ast=(\openone-A)^{-1}c$, $r_\ast:=\|v^\ast\|_2$. One notational
warning: \cite{PaperI} writes $r(A)$ for the \emph{spectral radius} of the Bloch
matrix; we write $\rho(A)$ for that quantity throughout, and reserve $r_\ast$
for the Bloch radius of the fixed point.

\subsection{Three lemmas}

\begin{lemma}[Choi--Frobenius isometry]\label{lem:iso}
For channels $\Phi,\Phi':M_d\to M_d$,
\begin{equation}
  \bigl\|J(\Phi)-J(\Phi')\bigr\|_F \;=\; \tfrac1d\,\|\Phi-\Phi'\|_{\mathrm{HS}} .
\end{equation}
For $d=2$ in Bloch form this reads
\begin{equation}\label{eq:iso2}
  \bigl\|J(\Phi)-J(\Phi')\bigr\|_F
  \;=\; \tfrac12\sqrt{\|A-A'\|_F^{2}+\|c-c'\|_2^{2}} .
\end{equation}
\end{lemma}

\begin{proof}
Put $\Delta=\Phi-\Phi'$. The blocks of
$J(\Phi)-J(\Phi')=\frac1d\sum_{ij}\ket{i}\bra{j}\otimes\Delta(\ket{i}\bra{j})$
are mutually orthogonal, so
$\|J(\Phi)-J(\Phi')\|_F^2=\frac1{d^2}\sum_{ij}\|\Delta(\ket{i}\bra{j})\|_F^2$.
The matrix units $\{\ket{i}\bra{j}\}$ form an orthonormal basis of $M_d$, so
that sum is $\|\Delta\|_{\mathrm{HS}}^2$. For $d=2$, in the orthonormal basis
$\{\openone/\sqrt2,\sigma_1/\sqrt2,\sigma_2/\sqrt2,\sigma_3/\sqrt2\}$ a
trace-preserving map has superoperator matrix
$\bigl(\begin{smallmatrix}1&0\\ c&A\end{smallmatrix}\bigr)$; hence $\Delta$ has
matrix $\bigl(\begin{smallmatrix}0&0\\ \Delta c&\Delta A\end{smallmatrix}\bigr)$
and $\|\Delta\|_{\mathrm{HS}}^2=\|\Delta A\|_F^2+\|\Delta c\|_2^2$.
\end{proof}

Note that Eq.~\eqref{eq:iso2} is an \emph{identity}, and that trace preservation
is essential: for maps of unequal trace the missing $(0,0)$ entry contributes
$|\Delta t|^2$ and the equality fails.

\begin{lemma}[Fixed-point collapse]\label{lem:collapse}
Let $\Phi$ be a qubit channel with $\|A\|_2<1$. Then $\Phi^n$ has Bloch pair
$(A^n,c_n)$ with $c_n=(\openone+A+\dots+A^{n-1})c$, and
\begin{equation}
  c_n-v^\ast \;=\; -\,A^{n}v^\ast .
\end{equation}
\end{lemma}

\begin{proof}
$c_n=(\openone-A)^{-1}(\openone-A^n)c=(\openone-A^n)(\openone-A)^{-1}c
 =(\openone-A^n)v^\ast$, the two factors commuting as functions of $A$.
\end{proof}

Thus the entire deviation of $\Phi^n$ from its limit is carried by powers of
$A$: no separate estimate of the stationary-state drift is needed.

\begin{lemma}[Decay of the Choi matrix]\label{lem:decay}
Let $\Phi:M_d\to M_d$ be a channel with fixed state $\rho^\ast$ and $\eta<1$,
and let $\Phi_\infty(X)=\operatorname{Tr}(X)\,\rho^\ast$, so that
$J(\Phi_\infty)=\tfrac{\openone}{d}\otimes\rho^\ast$. Then
\begin{equation}\label{eq:decay}
  \bigl\|J(\Phi^n)-\tfrac{\openone}{d}\otimes\rho^\ast\bigr\|_F
  \;\le\; \frac{\eta^{n}}{d}\,
  \sqrt{\,d\,\bigl\|\tfrac{\openone}{d}-\rho^\ast\bigr\|_F^{2}+d^{2}-1\,}.
\end{equation}
For $d=2$ this is $\tfrac12 a^{n}\sqrt{r_\ast^{2}+3}$.
\end{lemma}

\begin{proof}
Take the orthonormal basis $\{\openone/\sqrt d\}\cup\{F_\alpha\}_{\alpha=1}^{d^2-1}$
of $M_d$ with every $F_\alpha$ traceless. Since $\Phi^n(\rho^\ast)=\rho^\ast$,
\begin{equation}
 (\Phi^n-\Phi_\infty)(\openone/\sqrt d)=\sqrt d\,\bigl[\Phi^n(\openone/d)-\rho^\ast\bigr]
 =\sqrt d\,\Phi^n\!\bigl(\tfrac{\openone}{d}-\rho^\ast\bigr),
\end{equation}
and $\tfrac{\openone}{d}-\rho^\ast$ is traceless, so this has norm at most
$\sqrt d\,\eta^{n}\|\tfrac{\openone}{d}-\rho^\ast\|_F$. On each traceless
$F_\alpha$, $(\Phi^n-\Phi_\infty)(F_\alpha)=\Phi^n(F_\alpha)$ has norm at most
$\eta^n$. Summing squares gives
$\|\Phi^n-\Phi_\infty\|_{\mathrm{HS}}^2\le\eta^{2n}\bigl[d\|\tfrac{\openone}{d}-\rho^\ast\|_F^2+d^2-1\bigr]$,
and Lemma~\ref{lem:iso} supplies the factor $1/d$. For $d=2$,
$\|\tfrac{\openone}{2}-\rho^\ast\|_F=r_\ast/\sqrt2$ and $d^2-1=3$.
\end{proof}

\subsection{The qubit bound}

\begin{theorem}\label{thm:qubit}
Let $\Phi$ be a qubit channel with Bloch pair $(A,c)$, $a=\|A\|_2<1$, and
fixed state $\rho^\ast$ with Bloch vector $v^\ast$, $r_\ast=\|v^\ast\|<1$. Then
\begin{equation}\label{eq:thm1}
  n_{\mathrm{EB}}(\Phi)\;\le\;
  \left\lceil
   \frac{\ln\!\bigl(2\sqrt{r_\ast^{2}+3}\,/\,(1-r_\ast)\bigr)}{\ln(1/a)}
  \right\rceil .
\end{equation}
More sharply,
\begin{equation}\label{eq:thm1sharp}
  n_{\mathrm{EB}}(\Phi)\le\min\Bigl\{n:\
    \sqrt{\|A^{n}\|_F^{2}+\|A^{n}v^\ast\|^{2}}\le\tfrac{1-r_\ast}{2}\Bigr\}.
\end{equation}
\end{theorem}

\begin{proof}
$J(\Phi_\infty)=\tfrac{\openone}{2}\otimes\rho^\ast$ is invariant under partial
transposition of the reference factor, and its spectrum is
$\{(1\pm r_\ast)/4\}$, so $\lambda_{\min}[J(\Phi_\infty)^{T_R}]=(1-r_\ast)/4>0$. Partial
transposition preserves $\|\cdot\|_F$, and $\|\cdot\|_\infty\le\|\cdot\|_F$, so
Weyl's inequality together with Lemma~\ref{lem:decay} gives
\begin{equation}\label{eq:weylchain}
  \lambda_{\min}\bigl[J(\Phi^{n})^{T_R}\bigr]
  \;\ge\; \frac{1-r_\ast}{4}-\frac{a^{n}}{2}\sqrt{r_\ast^{2}+3}.
\end{equation}
The right-hand side is nonnegative as soon as
$a^{n}\le(1-r_\ast)\big/\bigl(2\sqrt{r_\ast^{2}+3}\bigr)$, which is
Eq.~\eqref{eq:thm1}. Then $\Phi^n$ is PPT, and since $d=2$ the
Peres--Horodecki criterion \cite{Peres96,Horodecki96} makes it entanglement
breaking. The sharper form follows by using $\|A^n\|_F$ and $\|A^nv^\ast\|$ in
Lemma~\ref{lem:decay} instead of $a^n\sqrt{r_\ast^2+3}$, via
Lemma~\ref{lem:collapse}.
\end{proof}

The hypothesis $r_\ast<1$ says exactly that $\rho^\ast$ has full rank; it cannot be
dropped, since $\Phi(X)=\operatorname{Tr}(X)\ket{0}\bra{0}$ has $r_\ast=1$ and makes
the bound vacuous (although that channel is trivially EB).

\subsection{Every dimension: the mechanism is not a qubit accident}

At $d=2$ the argument leaned on $\mathrm{EB}=\mathrm{PPT}$. It need not.

\begin{theorem}\label{thm:generald}
Let $\Phi:M_d\to M_d$ be a channel with $\eta<1$ and fixed state $\rho^\ast$ of
full rank, $p:=\lambda_{\min}(\rho^\ast)>0$. Put
\begin{equation}\label{eq:Kd}
  K_d \;:=\; \frac{\sqrt{d^{2}-1}}{p\,d}\,
    \sqrt{\,d\bigl\|\tfrac{\openone}{d}-\rho^\ast\bigr\|_F^{2}+d^{2}-1\,}\;.
\end{equation}
Then
\begin{equation}\label{eq:thm2}
  n_{\mathrm{EB}}(\Phi)\;\le\;
  \left\lceil \frac{\ln K_d}{\ln(1/\eta)} \right\rceil \;<\;\infty .
\end{equation}
\end{theorem}

\begin{proof}
Since $\rho^\ast\ge p\openone$, the operator $\rho^\ast-p\openone$ is positive
with trace $1-pd$, so $\rho^\ast=pd\,\tfrac{\openone}{d}+(1-pd)\,\omega$ for a
state $\omega$ (the case $pd=1$, i.e. $\rho^\ast=\openone/d$, is immediate).
Hence, with $D=d^{2}$,
\begin{equation}
  J(\Phi_\infty)=\tfrac{\openone}{d}\otimes\rho^\ast
  = pd\,\frac{\openone}{D}+(1-pd)\,\Bigl(\tfrac{\openone}{d}\otimes\omega\Bigr),
\end{equation}
the second term separable. Let $X$ be any state and
$\Delta=X-J(\Phi_\infty)$, so $\operatorname{Tr}\Delta=0$ and
\begin{equation}
  X = pd\Bigl[\frac{\openone}{D}+\frac{\Delta}{pd}\Bigr]
      +(1-pd)\Bigl(\tfrac{\openone}{d}\otimes\omega\Bigr).
\end{equation}
By the Gurvits--Barnum separable-ball theorem \cite{GurvitsBarnum02}, every
state within Frobenius distance $1/\sqrt{D(D-1)}$ of $\openone/D$ is separable;
so if $\|\Delta\|_F\le pd/\sqrt{D(D-1)}=p/\sqrt{d^{2}-1}$ then the bracket is a
separable state and $X$, a convex combination of separable states, is separable.
Applying this to $X=J(\Phi^{n})$ and inserting Eq.~\eqref{eq:decay} gives
$\eta^{n}\le 1/K_d$, which is Eq.~\eqref{eq:thm2}.
\end{proof}

Theorem~\ref{thm:generald} is logically independent of \cite{Park2026}: it
assumes contraction and says nothing about PPT-ness, whereas Theorem 1.1 of
\cite{Park2026} assumes PPT and says nothing about contraction. Together they
cover two disjoint mechanisms by which a quantum loop must eventually stop
carrying entanglement. At $d=2$ the two routes may be compared directly:
Theorem~\ref{thm:generald} tolerates
$\|\Delta\|_F\le(1-r_\ast)/\sqrt{12}$ against $(1-r_\ast)/4$ for
Theorem~\ref{thm:qubit}, so the separable-ball route is sharper there by the
factor $2/\sqrt3\simeq1.155$; the loss in Theorem~\ref{thm:qubit} is the
step $\|\cdot\|_\infty\le\|\cdot\|_F$.

\subsection{The unital case is exact, and fixes the constant}

\begin{proposition}\label{prop:unital}
A unital qubit channel is entanglement breaking if and only if
$\|A\|_1\le1$. Consequently
$n_{\mathrm{EB}}(\Phi)=\min\{n\ge1:\|A^{n}\|_1\le1\}$.
\end{proposition}

\begin{proof}[Proof sketch]
Rotations of the input and output Bloch balls are unitary conjugations and
preserve EB, so $A$ may be taken diagonal, $A=\mathrm{diag}(\lambda_1,\lambda_2,\lambda_3)$
with $|\lambda_i|$ its singular values \cite{RSW02}. The four eigenvalues of
$J(\Phi)^{T_R}$ are $\tfrac14(1\pm\lambda_1\pm\lambda_2\pm\lambda_3)$ with an
\emph{odd} number of minus signs; an even number gives the spectrum of
$J(\Phi)$ itself, i.e. the complete-positivity condition. Each parity class
cuts out a tetrahedron in $(\lambda_1,\lambda_2,\lambda_3)$, and the two
tetrahedra intersect in the octahedron $\sum_i|\lambda_i|\le1$. Complete
positivity is therefore essential: $(\lambda_i)=(0.9,0.9,-0.9)$ satisfies all
four partial-transpose inequalities while $\|A\|_1=2.7$. Imposing both,
$\Phi$ is PPT iff $\|A\|_1\le1$, and at $d=2$ PPT is EB.
\end{proof}

Note $\|A^n\|_1\ne\sum_i s_i(A)^n$ unless $A$ is normal, the same
non-normality that separates $\|A\|_2$ from $\rho(A)$ below.

\begin{corollary}[Tightness]\label{cor:tight}
For isotropic unital channels $A=aR$ with $R\in SO(3)$,
Proposition~\ref{prop:unital} gives $n_{\mathrm{EB}}=\lceil\ln3/\ln(1/a)\rceil$
whereas Theorem~\ref{thm:qubit} gives $\lceil\ln(2\sqrt3)/\ln(1/a)\rceil$. The
two differ, before rounding, by the constant factor
\begin{equation}\label{eq:tightconst}
  \frac{\ln\bigl(2\sqrt3\bigr)}{\ln 3}\;=\;1.1309\ldots
\end{equation}
\end{corollary}

So, before rounding, Theorem~\ref{thm:qubit} overshoots the exact answer by
$13.1\%$; the integer ceiling can of course cost more when $n_{\mathrm{EB}}$
is small (at $a=1/3$ the exact index is $1$ and the bound is $2$).

\begin{remark}\label{rem:gbtight}
Theorem~\ref{thm:generald} does better still on this family. At $r_\ast=0$ one has $\rho^\ast=\openone/2$, $p=1/2$, and
$K_2=\sqrt3\cdot\sqrt3=3$ exactly, so Eq.~\eqref{eq:thm2} returns
$\lceil\ln3/\ln(1/a)\rceil$: \emph{the exact answer}. The Gurvits--Barnum
radius is saturated by isotropic states of two qubits, which is precisely why
the separable-ball route loses nothing here while the Weyl route loses the
factor in Eq.~\eqref{eq:tightconst}.
\end{remark}

\begin{corollary}[One clock, not two]\label{cor:clock}
As $a\to1^-$, $\ln(1/a)=(1-a)\bigl(1+O(1-a)\bigr)$, so
\begin{equation}
  n_{\mathrm{EB}}(\Phi)\;\le\;
  \frac{\ln\bigl(2\sqrt{r_\ast^2+3}/(1-r_\ast)\bigr)}{1-\|A\|_2}\,\bigl(1+o(1)\bigr).
\end{equation}
For a \emph{unital} channel, Proposition~\ref{prop:unital} gives
$n_{\mathrm{EB}}=\min\{n:\|A^n\|_1\le1\}$ exactly, so for normal $A$ with
$k_{\mathrm{eff}}=\#\{i:s_i=s_1\}$ the elementary bound $\|A^n\|_1\ge
k_{\mathrm{eff}}s_1^{\,n}$ gives
\begin{equation}
  n_{\mathrm{EB}}(\Phi)\;\ge\;\frac{\ln k_{\mathrm{eff}}}{\ln(1/s_1)}
  \;=\;\frac{\ln k_{\mathrm{eff}}}{1-s_1}\bigl(1+o(1)\bigr),\qquad s_1\to1^- .
\end{equation}
The number of rounds a loop remains entangling and the number of rounds it
remembers its input are therefore the same quantity, up to a factor depending
only on $r_\ast$ and on how many Bloch axes contract at comparable rates.
\end{corollary}

\section{The certified index of the feedback loop}
\label{sec:certified}

\subsection{The value}

Evaluating Eq.~\eqref{eq:indexformula} at the reference point
\eqref{eq:refpoint} gives
\begin{align}
  \lambda_{\min}\bigl[J(\Phi)^{T_{R}}\bigr]   &= -0.2052788757\ldots \notag\\
  \lambda_{\min}\bigl[J(\Phi^{2})^{T_{R}}\bigr] &= -0.0451949891\ldots \label{eq:lamvalues}\\
  \lambda_{\min}\bigl[J(\Phi^{3})^{T_{R}}\bigr] &= +0.0370854674\ldots \notag
\end{align}
so that $\Phi$ and $\Phi^{2}$ are NPT and $\Phi^{3}$ is PPT. Hence
\begin{equation}\label{eq:N3}
  \boxed{\;\nEB(\Phi)\;=\;3\;}
\end{equation}
The loop conveys entanglement through two feedback rounds and is entanglement
breaking from the third onwards. Note that
Theorem~\ref{thm:qubit}, evaluated with $\|A\|_{2}=0.71331\ldots$ and
$r=0.58506\ldots$, returns the bound $\nEB\le7$; the direct evaluation
sharpens this to $3$.

\subsection{Certification at the reference point}

The values \eqref{eq:lamvalues} are certified in $256$-bit ball
arithmetic~\cite{Arb,pythonflint} following the workflow of
\cite[App.~B]{PaperI}, with one deliberate strengthening: no eigenvalue solver
appears in the certificate path. Negativity is established by a \emph{witness}: a vector
$w$ with $\bra{w}J^{T_{R}}\ket{w}<0$ certifies
$\lambda_{\min}\le\bra{w}J^{T_{R}}\ket{w}/\|w\|^{2}<0$ regardless of how $w$
was obtained, and positivity by \emph{Sylvester's criterion}, all four leading
principal minors of $J(\Phi^{3})^{T_{R}}$ being certified positive. Both are
pure arithmetic and therefore exact under ball arithmetic. As a cross-check the
same three quantities are also computed with Arb's certified eigenvalue solver,
and independently at $60$ decimal digits by direct composition of the channel,
agreeing to $25$ significant digits.

\subsection{Certification on the square}

Extending Eq.~\eqref{eq:N3} from a point to the square of half-width
$s=0.0013\pi$ requires more care than it appears to, and the obvious method
fails.

Evaluating the whole square in interval arithmetic, entering the four angles
as balls of radius $s$, does not work. Each angle occurs many times inside $U$
and $U^{\dagger}$, and ball arithmetic cannot know that those occurrences are
the same number; the resulting dependency inflates the enclosure of
$\lambda_{\min}[J(\Phi^{2})^{T_{R}}]$ to width $\approx0.13$, against a margin
of $0.045$. Uniform bisection restores tightness only at half-width
$\approx8\times10^{-5}\pi$, i.e. $16^{4}=65\,536$ boxes: brute force, not a
certificate. Nor does the generic perturbation constant
of~\cite{PaperI} suffice: $\|A-A_{0}\|_{2}\le12s$ is an $n=1$ bound, already
known to overestimate the true local variation by a factor of about five
\cite[Rem.~2]{PaperI}, and pushed through two powers it yields a drift of
$0.078$, larger than the margin it must beat.

What does work is the structured perturbation analysis called for in
\cite[Sec.~X]{PaperI}: a mean-value bound on the four parameter derivatives,
which is immune to the dependency problem because a derivative enclosure need
only be
crude: it is multiplied by $s$. Bounding
$\sum_{i}\sup\|\partial A/\partial p_{i}\|_{F}$ and
$\sum_{i}\sup\|\partial\bm{c}/\partial p_{i}\|$ over the square gives
\begin{equation}\label{eq:deltas}
  \|A-A_{0}\|_{F}\le\delta_{A}=0.0178,\qquad
  \|\bm{c}-\bm{c}_{0}\|\le\delta_{c}=0.0068 ,
\end{equation}
and these propagate through the exact identity of Lemma~\ref{lem:iso},
together with $\|A^{n}-A_{0}^{n}\|_{F}\le\sum_{k}\|A\|_{2}^{k}\delta_{A}
\|A_{0}\|_{2}^{\,n-1-k}$ and the corresponding bound for $\bm{c}_{n}$, to a
certified drift of $\lambda_{\min}$ across the square. Appendix~\ref{app:cert}
gives the details. The outcome is
\begin{equation}\label{eq:drifts}
  \begin{aligned}
    n=2:&\quad \text{drift}\le0.0153 \;<\; 0.045\le|\lambda_{\min}| ,\\
    n=3:&\quad \text{drift}\le0.0194 \;<\; 0.037\le \lambda_{\min} ,
  \end{aligned}
\end{equation}
i.e. certified Lipschitz constants $L_{2}\le3.739$ and $L_{3}\le4.742$ for
$\lambda_{\min}$ in the four angles. Weyl's inequality then gives $N=3$ at
every point of the square:
\begin{equation}
  \nEB(\Phi_{p})=3\qquad\text{for all } p\in\mathcal{B},
\end{equation}
i.e. on the whole box $\mathcal{B}$ of Eq.~\eqref{eq:box}, in which $\kappa$ and
$\beta$ vary as well, a strictly larger region than the $(\theta,\varphi)$
square of Proposition~4 of~\cite{PaperI}.

\begin{remark}
The certificate is conservative rather than tight. Substituting wider squares,
it continues to hold at $s=0.0013\pi$ and stops certifying at $0.004\pi$,
whereas direct evaluation shows $N=3$ persists to $\approx0.01\pi$ and first
breaks at $0.03\pi$, where $N\in\{2,3,4\}$. As with the four properties of
\cite[Rem.~2]{PaperI}, the certified square is a proof of principle, not a
phase boundary.
\end{remark}

\section{The index landscape}
\label{sec:landscape}

Figure~\ref{fig:landscape}(a) maps $N$ over the same $(\theta,\varphi)$ slice
at fixed $(\kappa_{0},\beta_{0})$ used for the four panels of Fig.~2
of~\cite{PaperI}, on an
$81\times81$ grid. The index is far from constant: it attains $93$ distinct
values, ranging from $1$ to $199$ among the points we resolve, with $38.8\%$
of the $6561$ grid points already entanglement breaking at $n=1$ and a ridge of long-lived
loops along small $\varphi$ and large $\theta$. The iteration was cut off at
$n=200$; $18$ of the $6561$ grid points reach that cutoff without becoming PPT.
They are shown in grey in panel~(a) and omitted from panel~(b) and from the
statistics below. The certified square sits in the $N=3$ contour. The $N=1$ region of panel~(a)
is the entanglement-breaking phase whose extent \cite[Sec.~X]{PaperI} leaves
open; we map it here but do not certify it, and certifying its boundary with
the same interval machinery remains open.

The structure is not accidental. Plotting the same data against the spectral
stability gap $1-\rho(A)$ (panel (b)) collapses it onto a single curve. Over
all $3997$ sampled channels with $1<N\le200$,
\begin{equation}\label{eq:oneclock}
  N\,\bigl[1-\rho(A)\bigr]\;\in\;[0.60,\,1.50],
  \qquad \text{median } 1.047,
\end{equation}
with $\ln3=1.0986$ sitting close to the median rather than at either edge.
Corollary~\ref{cor:clock} accounts for the scale and
Proposition~\ref{prop:unital} for the spread. In the unital case the index is
exactly $\min\{n:\|A^{n}\|_{1}\le1\}$, so the product is set by how many Bloch
axes contract at comparable rates: three comparable axes give $\ln3$, two give
$\ln2$, and one gives a product tending to zero. The data bear this out. The
ratio $s_{2}/s_{1}$ of the two largest singular values of $A$ correlates with
$N[1-\rho(A)]$ at $0.65$, and its median rises from $0.51$ among the points
below $\ln3$ to $0.84$ among those above; the corresponding soft count
$\sum_{i}(s_{i}/s_{1})^{N}$ correlates at $0.53$. The channels here are neither
unital nor normal, so Proposition~\ref{prop:unital} is a guide rather than a
law, but it is evidently the right guide.

The operational reading is the point of the section. The stability gap measures
how fast the loop forgets its input; $N$ measures how long it stays quantum.
Equation~\eqref{eq:oneclock} says these are not independent design parameters.
A loop cannot be made to forget quickly and remain entangling for many rounds:
buying a large stability margin is buying a short quantum lifetime, at a rate
set by $\ln3$ to within the factor of about $2.5$ spanned by
Eq.~\eqref{eq:oneclock}.

\begin{figure}[t]
  \includegraphics[width=\textwidth]{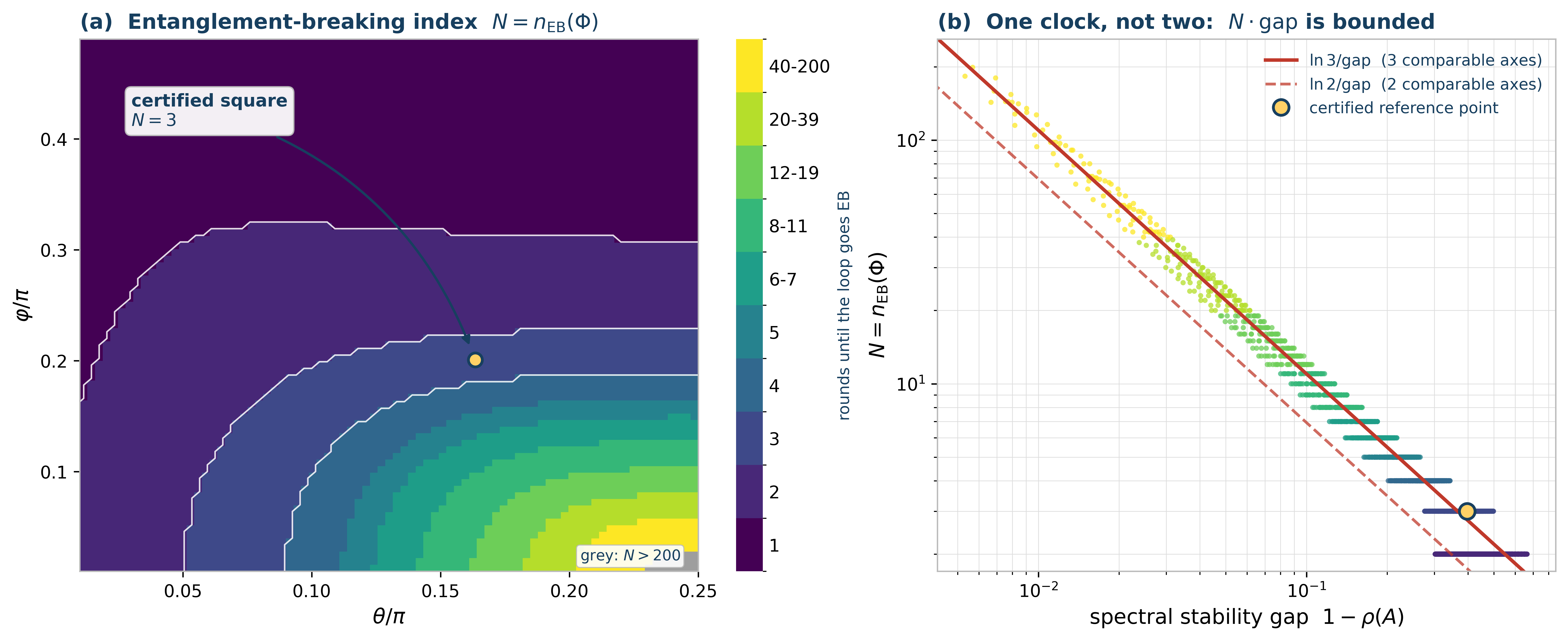}
  \caption{\label{fig:landscape}%
  (a) The entanglement-breaking index $N=\nEB(\Phi)$ over the
  $(\theta,\varphi)$ slice at fixed $(\kappa_{0},\beta_{0})$, computed from
  Eq.~\eqref{eq:indexformula}. White contours mark $N=1,2,3$; the marker is the
  certified reference point, where $N=3$ uniformly on the square of half-width
  $0.0013\pi$. (b) The same data against the spectral stability gap
  $1-\rho(A)$. The solid curve is the isotropic unital law of
  Proposition~\ref{prop:unital} in its $\rho\to1$ form, $\ln3/[1-\rho(A)]$, for
  three comparable Bloch axes; the dashed curve $\ln2/[1-\rho(A)]$ is the
  two-axis case. Grey cells in (a) did not become PPT within $200$ rounds. The product $N[1-\rho(A)]$ stays within
  $[0.60,1.50]$ over all $3997$ sampled channels with $1<N\le200$: the loop's
  memory time and
  its quantum lifetime are one clock, not two.}
\end{figure}

\section{The architecture floor}\label{sec:floor}

The central structural fact is that the infinite-temperature limit erases the bath entirely.

\begin{theorem}[Floor]\label{thm:floor}
For every traceless Hermitian involution $S$, the channel $\Phi_0$ is independent of $S$. Consequently $\nEB(\Phi_0)$, and indeed every functional of $\Phi_0$, is determined by the interaction unitary $U$ alone.
\end{theorem}

\begin{proof}
$\tau_S(0) = \tfrac12(\I + 0\cdot S) = \I/2$ for every $S$. The ancilla state is therefore $\I_4/4$ irrespective of the polarisation axis, and $\Phi_0(\rho) = \Tr_{FL}[U(\rho \otimes \I_4/4)U^{\dagger}]$ contains no reference to $S$.
\end{proof}

The proof is immediate, but the content is not. It says that at infinite bath temperature the channel retains no memory of the bath's structure, only of the circuit. We accordingly define the \textbf{architecture floor}
\begin{equation}\label{eq:floordef}
  \floor(U) \;:=\; \nEB(\Phi_0),
\end{equation}
a functional of the interaction alone. Certification confirms the identity at exact arithmetic: across five bath directions, including two off the coordinate axes, the resulting Bloch data differ by at most $6.3\times10^{-76}$ with every entry containing zero; the enclosures are bit-identical, as an exact identity requires (\S\ref{sec:numerics}).

The floor is not a technicality of the construction; it is frequently the whole story: for a large fraction of circuits the loop destroys entanglement in a single round even at infinite bath temperature, decided entirely by the interaction, before temperature is mentioned. Theorem~\ref{thm:normalform} makes that fraction computable.

\subsection{The infinite-temperature channel in closed form}\label{sec:normalform}

Theorem~\ref{thm:floor} says $\Phi_0$ depends only on $U$. It says nothing about \emph{how}. In fact $\Phi_0$ collapses to a three-parameter normal form, and one of the four circuit angles disappears entirely.

\begin{theorem}[Infinite-temperature normal form]\label{thm:normalform}
Let $U$ be as in \eqref{eq:U}, with independent branch angles $a_0,a_1$. Then
\[
  \Phi_0 \;=\; \Lambda_{\cos^2\varphi}\;\circ\;\mathrm{Ad}_{R_y(\beta)}\;\circ\;\Delta_\delta,
  \qquad \delta \;=\; \cos\!\big(\tfrac{a_0-a_1}{2}\big),
\]
where $\Delta_\delta$ is the dephasing channel with Bloch matrix $\diag(\delta,\delta,1)$ and $\Lambda_a$ is the isotropic contraction $r\mapsto ar$. Equivalently, in the Bloch basis $(x,y,z)$,
\[
  A_0 \;=\; \cos^2\!\varphi\;O_y(\beta)\,\diag(\delta,\delta,1), \qquad c_0 = 0 .
\]
In particular $\Phi_0$ is unital, and the coupling $\kappa$ does not appear. For the supplementary construction $\delta = \sin 2\theta$.
\end{theorem}

\begin{proof}
Write $c = \cos\varphi$, $s = \sin\varphi$, $S = \mathrm{SWAP}_{ML}$, $V_m = R_y(a_m)$, and $W = \exp(-i\tfrac\kappa2 Z\otimes Y)$ on $FL$, so $U_w = \I_M\otimes W$.

$R_y(\beta)_M$ commutes with $\Tr_{FL}$, so $\Phi_0 = \mathrm{Ad}_{R_y(\beta)}\circ\tilde\Phi_0$ with $\tilde\Phi_0$ built from $U' = U_fU_wU_W$; on Bloch data the conjugation is $O_y(\beta)$. Put $X := U_wU_W(\rho\otimes\tfrac{\I_4}{4})U_W^\dagger U_w^\dagger$. Since $U_f = c\I - isS$,
\begin{equation}\label{eq:split}
  U'(\rho\otimes\tfrac{\I_4}{4})U'^\dagger \;=\; c^2X + s^2SXS + ics\,[X,S].
\end{equation}

\emph{(i) The first term is a dephasing channel, and $\kappa$ cancels.} The partial trace over $FL$ is invariant under unitaries supported on $FL$, so the $W$'s drop out. With
$U_W(\rho\otimes\tfrac{\I_4}{4})U_W^\dagger = \tfrac14\sum_{m,m'}\rho_{mm'}|m\rangle\!\langle m'|\otimes V_mV_{m'}^\dagger\otimes\I_L$,
tracing out $F$ and $L$ multiplies the $(m,m')$ entry by $\tfrac12\Tr[V_mV_{m'}^\dagger]$, which is $1$ on the diagonal and
\[
  \tfrac12\Tr[V_0V_1^\dagger] \;=\; \tfrac12\Tr\big[R_y(a_0-a_1)\big] \;=\; \cos\!\big(\tfrac{a_0-a_1}{2}\big) \;=\; \delta
\]
off it. Hence $\Tr_{FL}[X] = \Delta_\delta(\rho)$.

\emph{(ii) The swapped term is constant.} $\Tr_{FL}[SXS]$ is the $L$-marginal of $X$ transplanted onto $M$. Since $\Tr_M[U_W(\rho\otimes\tfrac{\I_4}{4})U_W^\dagger] = \tfrac14\sum_m\rho_{mm}V_mV_m^\dagger\otimes\I_L = \tfrac14\I_F\otimes\I_L$, conjugating by the unitary $W$ and tracing out $F$ gives $\Tr_{FL}[SXS] = \tfrac12\I$, independent of $\rho$: this term contributes to the affine part only.

\emph{(iii) The cross term vanishes identically.} For $B$ on $F$ put $\mathcal G(B) := W(B\otimes\I_L)W^\dagger$. Using $(Z\otimes Y)^2 = \I$,
\[
  \mathcal G(B) = \cos^2\!\tfrac\kappa2\,B\otimes\I + \sin^2\!\tfrac\kappa2\,ZBZ\otimes\I + i\cos\tfrac\kappa2\sin\tfrac\kappa2\,(BZ-ZB)\otimes Y,
\]
and since $\Tr[ZBZ] = \Tr[B]$ while $\Tr[BZ-ZB] = 0$,
\begin{equation}\label{eq:Gtrace}
  \Tr_F\big[\mathcal G(B)\big] = \Tr[B]\,\I_L \qquad\text{for every }B.
\end{equation}
Writing $X = \tfrac14\sum_{m,m'}\rho_{mm'}|m\rangle\!\langle m'|\otimes\mathcal G(V_mV_{m'}^\dagger)$ and using $S|j,f,l\rangle = |l,f,j\rangle$, identity \eqref{eq:Gtrace} gives
\[
  \big(\Tr_{FL}[XS]\big)_{ij} = \tfrac14\rho_{ij}\Tr[V_iV_j^\dagger] = \big(\Tr_{FL}[SX]\big)_{ij},
\]
so $\Tr_{FL}[[X,S]] = 0$ identically.

Applying $\Tr_{FL}$ to \eqref{eq:split} therefore gives $\tilde\Phi_0(\rho) = c^2\Delta_\delta(\rho) + s^2\tfrac{\I}{2}$, which is unital, so $\tilde c_0 = 0$ and $\tilde A_0 = c^2\diag(\delta,\delta,1)$.
\end{proof}

$\kappa$ enters \eqref{eq:split} only through $W$, and $W$ cancels in (i), contributes a constant in (ii), and is annihilated by \eqref{eq:Gtrace} in (iii). \textbf{At infinite bath temperature the $F$--$L$ coupling is invisible.} This is the $p=0$ face of the same structure that Lemma~\ref{lem:structuralzero} exhibits at all $p$: there the supplementarity of the branch angles forces $A_{yz}=0$; here the whole $y$ row and column decouple and $\kappa$ vanishes with them.

\subsection{Three consequences}\label{sec:consequences}

\paragraph{(a) The near-unitarity parameter is exact.} The singular values of $A_0$ are
\[
  \big\{\,\cos^2\varphi,\;\; \cos^2\varphi\,|\delta|,\;\; \cos^2\varphi\,|\delta|\,\big\},
\]
one simple and one doubly degenerate, so $\opn{A_0} = \cos^2\varphi$ and
\begin{equation}\label{eq:epsilon}
  \varepsilon \;:=\; 1 - \opn{A_0} \;=\; \sin^2\!\varphi
\end{equation}
exactly, for every $\theta,\kappa,\beta$. The scan of \S\ref{sec:valley} bins on $\varepsilon$; by this identity those bins are exact slices of parameter space rather than measured quantities.

\paragraph{(b) A closed form for the floor.} Both $\mathrm{span}\{e_y\}$ and the $(x,z)$ plane are invariant under $A_0$, since $O_y(\beta)$ fixes $e_y$ and $\diag(\delta,\delta,1)$ is diagonal. Writing $a = \cos^2\varphi$ and $M = a\,\Rot(\beta)\diag(\delta,1)$ for the action on that plane,
\begin{equation}\label{eq:nucAn}
  A_0^n = (a\delta)^n \oplus M^n, \qquad
  \nuc{A_0^n} = (a|\delta|)^n + \sqrt{\Frob{M^n}^2 + 2|\det M^n|},
\end{equation}
the second equality because a real $2\times2$ matrix has $\sigma_1\sigma_2 = |\det|$ and $\sigma_1^2+\sigma_2^2 = \Frob{\cdot}^2$. Cayley--Hamilton gives $M^n = p_nM + q_n\I$ with
\[
  \mu_\pm = \tfrac12\Big[\Tr M \pm \sqrt{(\Tr M)^2 - 4\det M}\Big],\quad
  p_n = \frac{\mu_+^n-\mu_-^n}{\mu_+-\mu_-},\quad q_n = -\det M\cdot p_{n-1},
\]
whence $\Frob{M^n}^2 = p_n^2\,\Frob{M}^2 + 2p_nq_n\Tr M + 2q_n^2$, with $\Tr M = a(1+\delta)\cos\beta$, $\det M = a^2\delta$ and $\Frob{M}^2 = a^2(1+\delta^2)$. Since $\Phi_0^n$ is unital for every $n$, the nuclear-norm characterisation of entanglement breaking for unital qubit channels \cite{Ruskai2003} (see also \cite{EBnorm}) applies at every step and
\begin{equation}\label{eq:floorclosed}
  \floor(U) \;=\; \min\{\,n\ge1 : \nuc{A_0^n}\le 1\,\}
\end{equation}
is a closed-form arithmetic condition in $(\varphi,\delta,\beta)$: no matrix powers of $A_0$, no singular-value decomposition, no Choi construction. It reproduces every one of the $7{,}056$ floors computed independently by Choi-matrix bisection in the scan of \S\ref{sec:valley}, and both counterexample floors of \S\ref{sec:valley} are certified directly from it in ball arithmetic (\S\ref{sec:numerics}).

At $n=1$ it reduces to $\nuc{A_0} = \cos^2\!\varphi\,(1+2|\delta|)$, so
\begin{equation}\label{eq:floorone}
  \floor = 1 \iff \cos^2\!\varphi\,(1+2|\sin2\theta|)\le1,
\end{equation}
independent of $\kappa$ \emph{and} of $\beta$.

\paragraph{(c) The floor distribution is a measure, not a sample.} Drawing $\theta,\varphi,\kappa,\beta$ independently and uniformly on $(0.02\pi,\,0.48\pi)$, with the support trimmed away from the $\varphi \to 0$ edge where the floor diverges, consequence (b) turns the floor-$=1$ fraction into a one-dimensional integral. On that support $\sin2\theta>0$ and the threshold $\varphi^*(\theta) = \arccos\big((1+2\sin2\theta)^{-1/2}\big)$ lies strictly inside the $\varphi$ range (both facts certified in \S\ref{sec:numerics}), so
\begin{equation}\label{eq:measure}
  \Pr\big[\floor = 1\big] \;=\; \frac{1}{\Delta^2}\int_{0.02\pi}^{0.48\pi}\!\big(0.48\pi - \varphi^*(\theta)\big)\,d\theta
  \;=\; 0.46427058204502844932,
\end{equation}
with $\Delta = 0.46\pi$, certified to a radius of $1.5\times10^{-21}$. The remainder of the distribution depends on $\beta$ as well and is evaluated on the same closed form over $10^7$ draws:

\begin{center}\footnotesize\setlength{\tabcolsep}{4pt}
\begin{tabular}{lcccccccccc}
\toprule
floor & 1 & 2 & 3 & 4 & 5 & 6--9 & 10--19 & 20--49 & 50--99 & $\ge100$\\
\midrule
prob. & \textbf{0.46427} & 0.23414 & 0.09993 & 0.05211 & 0.03169 & 0.05619 & 0.03612 & 0.01836 & 0.00497 & 0.00212\\
\bottomrule
\end{tabular}
\end{center}

\noindent Sampling standard errors are at most $1.6\times10^{-4}$, and the floor-$=1$ entry is the certified value \eqref{eq:measure} rather than its Monte Carlo estimate. The distribution has mean $3.669$ and median $2$, and a hard maximum of $279$ on this support, since $\nuc{A_0^n}\le3\cos^{2n}\!\varphi$ forces $n\le\ln3/(-\ln\cos^2 0.02\pi)$.

\section{The exchange-collision baseline}\label{sec:baseline}

Consider an excitation-exchange collision between the system qubit and a fresh thermal qubit ancilla, a Jaynes--Cummings or XY interaction. Its single-collision map is the generalised amplitude damping
channel (GADC) \cite{KSW2020}, with damping $\gamma$ and fixed point
$\tau(p)$. At fixed $p$ these channels are closed under composition, and the
index needs no certification at all.

\begin{proposition}[Closed-form index of the GADC]\label{prop:gadc}
For damping $0<\gamma<1$ and polarisation $0\le p<1$,
\begin{equation}\label{eq:gadcclosed}
  \nEB(\gamma,p) \;=\; \left\lceil \frac{\ln x_-(p)}{\ln(1-\gamma)} \right\rceil,
  \qquad x_-(p) \;=\; \frac{\sqrt{2-p^2}-1}{\sqrt{2-p^2}+1},
\end{equation}
exactly, including the case in which the ratio is an integer. At $p=1$ the
index is infinite.
\end{proposition}

\begin{proof}
The parameter region in which a single GADC is entanglement breaking was
determined in \cite{KSW2020} (their noise parameter is $(1-p)/2$ in the present
notation). The $n$-fold composite is again a GADC, with damping
$1-(1-\gamma)^n$ and the same $p$, so the index follows from that region by
composition; we give the derivation for completeness. In Bloch form the
composite has $A_n=\diag(x^{1/2},x^{1/2},x)$ and $c_n=(0,0,(1-x)p)$ with
$x=(1-\gamma)^n$. For a Bloch-diagonal channel $\diag(a,a,b)$ with shift
$(0,0,c)$ the partial-transposed Choi state has spectrum
\[
  \tfrac14\big\{\,1+b+c,\;\; 1+b-c,\;\; (1-b)+\sqrt{c^2+4a^2},\;\;
  (1-b)-\sqrt{c^2+4a^2}\,\big\},
\]
and the first pair is nonnegative because $|c|=(1-b)p\le1-b\le1+b$. Here
$a^2=x=b$, so PPT, which for qubits is equivalent to entanglement breaking
\cite{HSR2003,Peres96,Horodecki96}, reduces to $(1-x)^2(1-p^2)\ge4x$. With
$q=1-p^2$ this reads $qx^2-(2q+4)x+q\ge0$, whose admissible root is
$x_-=(q+2-2\sqrt{q+1})/q=(\sqrt{q+1}-1)^2/q=(\sqrt{2-p^2}-1)/(\sqrt{2-p^2}+1)$.
Since $x$ decreases strictly in $n$ and PPT is a closed condition, the smallest
qualifying $n$ is the ceiling in \eqref{eq:gadcclosed}. At $p=1$ the condition
reads $0\ge4x$, which fails at every finite $n$.
\end{proof}

The formula was cross-validated on $500$ random $(\gamma,p)$ points against an
independent Choi-eigenvalue route, and certified in ball arithmetic at four
exact dyadic points with negative controls (Appendix~\ref{app:code}). Three
consequences matter below. First, $x_-$ is strictly decreasing in $|p|$, so
\textbf{within the exchange family the index is monotone in bath
polarisation}: colder never hurts, as a theorem. Second, the
infinite-temperature anchor is $x_-(0)=3-2\sqrt2$, giving
$\nEB=\lceil\ln(3+2\sqrt2)/|\ln(1-\gamma)|\rceil\approx1.7627/\gamma$ at weak damping.
Third, as $p\to1$ the index diverges, but only logarithmically,
$\nEB\sim\ln[4/(1-p^2)]/|\ln(1-\gamma)|$, and the divergence is inherited from
the purity of the fixed point: at $p=1$ the composite is pure amplitude
damping, which is never entanglement breaking below total damping. The
feedback loop has a full-rank fixed point and, across its certified
coupling-temperature surface, a thermal ratio bounded by $3/2$
(\S\ref{sec:bathterm}): the architecture regularises the pure-bath divergence.

Monotonicity in temperature is therefore guaranteed for plain exchange collisions and not for the structured interaction of the loop. Both are
collision models with fresh ancillas in every round (\S\ref{sec:model}), but
the loop is not a GADC, and its composites are not again members of its family:
the family is not a semigroup (\S\ref{sec:priorwork}). The next section shows that monotonicity then fails, rarely but certifiably. No structure, no valley.

\section{The floor is not the minimum}\label{sec:valley}

It is natural to expect that raising the bath temperature can only shorten the loop's quantum lifetime, so that $\Phi_0$ minimises the index over the thermal family. For plain exchange collisions this is a theorem (Proposition~\ref{prop:gadc}). For the feedback loop it is \textbf{false}.

\paragraph{Counterexample (certified).} At the circuit
\[
  (\theta,\varphi,\kappa,\beta)/\pi = (0.263681,\; 0.008960,\; 0.377631,\; 0.004911), \qquad \opn{A_0} = 0.999208,
\]
the infinite-temperature index is $\nEB(\Phi_0) = 352$, while $\nEB(\Phi_p) = 350,\ 347,\ 351$ at $p = 0.25,\ 0.50,\ 0.75$ respectively, i.e.\ a bath term (\S\ref{sec:bathterm}) of $-2,\ -5,\ -1$. A second circuit, $(\theta,\varphi,\kappa,\beta)/\pi = (0.219750, 0.022784, 0.033098, 0.024026)$ with $\opn{A_0} = 0.994885$, gives floor $67$ and $\nEB(\Phi_{0.25}) = 66$.

Every index above is certified in Arb ball arithmetic at 256-bit working precision, each on both boundaries: NPT by explicit witness at $n-1$, PPT by Sylvester's criterion at $n$ (\S\ref{sec:numerics}). These are not floating-point artefacts.

\textbf{A polarised bath can therefore break entanglement strictly sooner than a maximally mixed one.} The floor of \S\ref{sec:floor} is a functional of the interaction and remains well defined, but it is not a lower bound on the family. We call the phenomenon a \textbf{valley}, and the rest of this section characterises it quantitatively.

\subsection{How the valley was measured}

Because the effect is rare, its characterisation is statistical, and we label it as such throughout: the \emph{existence} claim rests on the certified counterexamples above, never on the sampling below.

Circuits were drawn with $\theta \sim \pi\big(\tfrac14 + \mathcal N(0,\,0.05^2)\big)$, $\kappa,\beta \sim \mathcal U(0,\tfrac\pi2)$, and $\varphi$ fixed within each bin by $\varepsilon = \sin^2\varphi$, an identity, not an approximation, by Theorem~\ref{thm:normalform}, so the bins are exact slices of the parameter space. A circuit is recorded as a valley if $\Phi_p^{\,\mathrm{floor}-1}$ is PPT for at least one $p$ on a 16-point grid; since entanglement breaking is closed under further composition, this single Choi evaluation per $p$ decides valley membership without an index search. The full index is then computed only for the circuits that pass this screen: at most $0.85\%$ of a bin, the largest valley rate observed.

Three independent runs contribute $4{,}114{,}799$ trials and $7{,}056$ valleys. The two runs sharing an $\varepsilon$ value at $3.0\times10^{-3}$ agree: $2{,}295/672{,}800 = 0.3411\%$ against $42/13{,}200 = 0.3182\%$, $z = +0.45$ (two-sided $P = 0.65$; capital $P$ denotes a $P$-value throughout, lower-case $p$ being the bath polarisation), across different sampling schemas and a 51-fold difference in exposure.

\subsection{\texorpdfstring{A Gaussian cutoff in $\varepsilon$, with a measured \emph{marginal} exponent}{A Gaussian cutoff in epsilon, with a measured marginal exponent}}\label{sec:cutoff}

The valley rate is flat deep in the near-unitary corner and falls off sharply as $\varepsilon$ grows. Over the five bins with $\varepsilon \le 1.5\times10^{-3}$ the rate is constant at
\[
  \text{plateau} \;=\; 0.7391\%\ \ [0.6744\%,\, 0.8083\%], \qquad \chi^2 = 3.56 \text{ on } 4 \text{ df},
\]
and by $\varepsilon = 8\times10^{-3}$ it has fallen by a factor of $383$ $[224, 718]$.

Fitting the binomial likelihood over all $11$ bins, with the runs pooled per $\varepsilon$, discriminates the candidate shapes decisively. Writing $\text{rate}(\varepsilon) = \mathcal{A}\exp[-(\varepsilon/\varepsilon_0)^\eta]$ and profiling in the exponent:
\begin{equation}\label{eq:beta}
  \eta \;=\; 2.02, \qquad 95\% \text{ CI } [1.90,\, 2.16].
\end{equation}
The interval contains $2$ and excludes $1$. A likelihood-ratio test of the free exponent against $\eta = 2$ gives $D = 0.126$ on 1 df ($P = 0.72$), no evidence whatever for a departure from quadratic, while against $\eta = 1$ it gives $D = 373.6$. A power law is worse than the Gaussian by $4957$ AIC units and a constant rate by $12219$; neither is tenable.

The fitted Gaussian is $\mathcal{A} = 0.7898\%$, $\varepsilon_0 = 3.290\times10^{-3}$, with half-maximum at $\varepsilon = 2.739\times10^{-3}$, the amplitude agreeing with the independently measured plateau to within its interval.

\textbf{This exponent is a marginal, and it should not be read as a property of the mechanism.} The scan's $\varepsilon$ parametrises the circuit \emph{ensemble}, not merely a coupling: the floor distribution shifts with it, from median floor $69$ at $\varepsilon = 2\times10^{-3}$ to median $32$ at $\varepsilon = 2\times10^{-2}$. So
\[
  \text{rate}(\varepsilon) \;=\; \sum_F \Pr[\floor = F \mid \varepsilon]\;\text{rate}(\varepsilon \mid \floor = F),
\]
and \eqref{eq:beta} measures the product. Resolving the same scan by floor (the tally records one) gives conditional exponents that are not $2$ and are not even close to each other:

\begin{center}\footnotesize
\begin{tabular}{@{}lrcl@{}}
\toprule
floor band & valleys & $\eta$ & 95\% CI\\
\midrule
$2$--$20$ (with \S\ref{sec:fartail}) & $33$ & $0.80$ & $[0.26,\,1.58]$\\
$33$--$64$ & $1{,}606$ & $2.70$ & $[2.35,\,3.10]$\\
$65$--$128$ & $2{,}892$ & $2.95$ & $[2.55,\,3.45]$\\
$\ge 129$ & $1{,}898$ & $3.60$ & $[3.00,\,4.40]$\\
\midrule
all floors, same bins & $6{,}538$ & $1.75$ & $[1.60,\,1.85]$\\
\bottomrule
\end{tabular}
\end{center}

\noindent The conditional exponent rises monotonically with the architecture floor, and \textbf{no band has an exponent of $2$}: each of the three deep bands excludes it, at $D = 13.6$, $26.4$ and $29.1$ on 1 df. The all-floors marginal on these six bins, $1.75$, sits inside the spread because it is an average over it, weighted by a floor law that is itself a function of $\varepsilon$. It differs from the $2.02$ of \eqref{eq:beta} because \eqref{eq:beta} additionally includes the first run's five plateau bins, which pin the amplitude; that the two intervals, $[1.60, 1.85]$ and $[1.90, 2.16]$, are \emph{disjoint} is the mixture speaking once more: no single law of this form describes the plateau bins and the tail bins simultaneously.

We therefore do \emph{not} claim a structural echo between the measured $2.02$ and the $O(p^2)$ behaviour forced by the structural zero $\Tr[K_nA_1] = 0$ of \S\ref{sec:obstruction}. Such an argument would need the conditional exponent, and the conditional exponent is not $2$. What \eqref{eq:beta} does establish is the law obeyed by a circuit drawn from this ensemble at a given $\varepsilon$, which is the operationally relevant statement and the one an experiment would meet; a mechanism must reproduce the ladder, not the average.

All rows above are statistics in the sense of \S\ref{sec:scanprov}. Their log-likelihoods were re-evaluated at 30 decimal digits, so none of the exclusions is an artefact of double precision (\S\ref{sec:scan3cert}).

\subsection{The far tail is a shallow-floor population with its own exponent}\label{sec:fartail}

One valley occurred at $\varepsilon = 2\times10^{-2}$, where the Gaussian above predicts $4.75\times10^{-13}$ events. Adding a single constant term repairs every bin simultaneously,
\begin{equation}\label{eq:tail}
  \text{rate}(\varepsilon) \;=\; \mathcal{A}\,e^{-(\varepsilon/\varepsilon_0)^2} \;+\; B, \qquad B = 7\times10^{-7}\ \ [4.1\times10^{-8},\, 3\times10^{-6}],
\end{equation}
improving the likelihood by $D = 38.76$ on 1 df and reducing the largest standardised residual across all $11$ bins to $1.37$, while moving the Gaussian component imperceptibly ($\mathcal{A} = 0.7902\%$, $\varepsilon_0 = 3.288\times10^{-3}$). $B = 0$ is excluded at $\Delta\log L = 19.4$.

The tail event's floor is $8$, the $3.8$th percentile of the floor distribution in its own bin, against 41st to 60th for the bulk bins. That percentile falls monotonically with $\varepsilon$ ($59.8$, $57.9$, $52.5$, $41.4$ at $\varepsilon = 2, 3, 5, 8 \times10^{-3}$), and $3.8$ is close to where that trend extrapolates, so the event is the end of a drift toward shallow floors rather than a separate population. \textbf{The drift is the phenomenon.} We resolved it directly.

\paragraph{A stratified scan, pre-registered.} Because Theorem~\ref{thm:normalform} makes the floor a $2\times2$ binary exponentiation, circuits can be screened on floor before any Choi matrix is built. A second scan drew $12{,}000{,}000$ circuits at $\varepsilon \in \{1.2, 2, 3.2, 5, 8, 12.5\}\times10^{-2}$, up to six times the largest $\varepsilon$ reached above, keeping the $7{,}362{,}197$ with $\floor \in [2,20]$, and found $23$ valleys. Two hypotheses were registered before it ran:
\begin{itemize}\itemsep2pt
\item $H_{\mathrm{gauss}}$: the conditional rate follows \eqref{eq:beta}, predicting $0.0064$ valleys in total;
\item $H_{\mathrm{flat}}$: the conditional rate is constant at the value this stratum shows above, $10/908064 = 1.101\times10^{-5}$, predicting $81.1$.
\end{itemize}
Both are refuted. Against $H_{\mathrm{gauss}}$, $P[X \ge 23] = 1.4\times10^{-73}$, and inflating its mean a thousandfold leaves $P = 3\times10^{-7}$. Against $H_{\mathrm{flat}}$, the exact binomial $P[X \le 23]$ is $2.650\times10^{-14}$, two-sided $5.300\times10^{-14}$. \textbf{A valley occurs at $\varepsilon = 0.125$}, floor $3$, depth $1$, where \eqref{eq:beta} predicts a rate below $10^{-300}$.

\paragraph{The stratum's law.} Pooling with the bins above (they agree where they overlap, Fisher $P = 1.00$ and $0.12$) gives $33$ valleys in $8{,}270{,}261$ screened circuits across a $62$-fold range in $\varepsilon$. Fitting $\mathcal{A}\exp[-(\varepsilon/\varepsilon_0)^\eta]$:
\begin{equation}\label{eq:stratumeta}
  \eta \;=\; 0.80, \qquad 95\% \text{ CI } [0.26,\, 1.58],
\end{equation}
consistent with $\eta = 1$ ($D = 0.36$, $P = 0.55$) and \textbf{excluding $\eta = 2$} ($D = 7.16$, $P = 0.0075$). An exponential is preferred over a Gaussian by $6.80$ AIC units, over a power law by $8.15$, and over a constant by $39.18$. The shallow stratum's characteristic scale, $\varepsilon_0 = 2.6\times10^{-2}$, is eight times that of \eqref{eq:beta}.

\paragraph{The decline is intrinsic, not composition drift.} Within $[2,20]$ the floor composition also moves with $\varepsilon$, so the same marginalisation could hide one level down. It does not. Weighting the pooled sub-band rates by each bin's own certified composition predicts a nearly flat rate; the observed rate spans a factor of $43$, with $\chi^2 = 74.7$ on 9 df ($P = 1.8\times10^{-12}$) and an observed-to-predicted ratio falling monotonically from $6.5$ to $0.08$. At $\varepsilon = 0.125$ composition alone predicts $11.9$ valleys and one was seen. The suppression is therefore real at fixed floor, just far gentler than \eqref{eq:beta}.

\paragraph{What $B$ is.} The stratum law at $\varepsilon = 2\times10^{-2}$ gives $5.2\times10^{-6}$, and the stratum carries $29.7\%$ of circuits there, so its marginal contribution is $1.6\times10^{-6}$, consistent with the fitted $B$. \textbf{$B$ is the shallow-floor population showing through once the deep-floor population has been extinguished.} It is not, however, a constant: it decays with $\varepsilon_0 \approx 2.6\times10^{-2}$, and looks flat across the bins of \eqref{eq:tail} only because that window is narrow compared with its own scale. Between $\varepsilon = 5\times10^{-3}$ and $0.125$ the fitted stratum law falls by a factor of $99$, and the measured rates by $43$.

The consequence for the paper's thesis is a strengthening. \textbf{For shallow architectures the floor fails to bound the family with no useful $\varepsilon$ window at all}: valleys survive to $\varepsilon = 0.125$, where the circuit is nowhere near unitary. The Gaussian window of \eqref{eq:beta} is a property of the deep, near-unitary corner alone.

\subsection{Depth is unit-quantised and sub-linear in the floor}

Define the depth of a valley as $\mathrm{floor} - \min_p \nEB(\Phi_p)$. Over the $6{,}538$ valleys from the two runs whose schema records depth (the first run predates that column; \S\ref{sec:scanprov}), it takes only small integer values: $83.82\%$ have depth $1$, $12.25\%$ depth $2$, $3.79\%$ depth $3$, and $0.14\%$ depth $4$. The largest \emph{relative} depths are all unit dips at small floors ($1/6$, $1/7$, $1/8$, $1/8$, $1/9$, $1/11$, $1/11$, $1/14$), so the relative scale is set by the denominator, not by unusually deep excursions.

Two natural mechanisms are both refuted. Fitting a Poisson regression $\mathbb E[\text{depth}] = C\,\mathrm{floor}^{\,b}$ gives
\begin{equation}\label{eq:depthglm}
  b \;=\; 0.239, \qquad 95\% \text{ CI } [0.200,\, 0.275],
\end{equation}
excluding $b = 0$ (a constant quantum, depth independent of the floor) and $b = 1$ (strict proportionality, a fixed relative depth) at $P < 10^{-31}$ and $P < 10^{-300}$ respectively. Depth grows with the floor, but far more slowly than proportionally.

The certified counterexample (floor $352$, depth $5$) is held out of this fit and is a $\sim2\sigma$ high outlier under it ($\mathbb E = 1.64$, $P(\text{depth}\ge5) = 0.026$), as expected of a circuit found by a search targeting deep valleys rather than by uniform sampling. It is consistent with the law and must not be pooled into it.

\subsection{The valley is a wide, contiguous, interior window in temperature}\label{sec:window}

Recomputing $\nEB(\Phi_p)$ on the full grid for all $6{,}538$ valley circuits, rather than reading the scan's recorded minimiser (which is biased to the lower edge of any plateau by its update rule), gives the geometry of the dip.

\begin{center}
\begin{tabular}{lr}
\toprule
sub-floor window contiguous in $p$ & \textbf{$6{,}538 / 6{,}538 = 100.00\%$}\\
median window width & 6 of 16 grid points, $\Delta p \approx 0.38$\\
windows resolved by a single grid point & $3.56\%$\\
median midpoint of the minimising set & $p = 0.4300$\\
$\nEB < \mathrm{floor}$ at $p = 1$ & $780$ $(11.9\%)$\\
$\nEB > \mathrm{floor}$ at $p = 1$ & $4{,}675$ $(71.5\%)$\\
\bottomrule
\end{tabular}
\end{center}

\textbf{Every valley is a single connected interval of bath polarisation}, typically spanning nearly forty per cent of the temperature axis, sitting in the interior, and closing again before $T = 0$, where in $71.5\%$ of cases the index has returned \emph{above} the floor. The dip is a mid-temperature phenomenon, not a zero-temperature one, and not a knife-edge numerical artefact. Its midpoint drifts warm as $\varepsilon$ falls: median $p = 0.208,\ 0.335,\ 0.430,\ 0.462$ at $\varepsilon = 8,\,5,\,3,\,2 \times10^{-3}$.

The certified counterexamples show exactly this shape. For the deep anchor the index runs
\[
  \footnotesize
  \begin{aligned}
    &352,\,352,\,352,\,351,\,351,\,350,\,349,\,349,\,348,\,348,\,\mathbf{347},\\
    &347,\,347,\,348,\,349,\,351,\,353,\,356,\,361,\,367,\,374
  \end{aligned}
\]
as $p$ runs from $0$ to $1$ in steps of $0.05$: a flat-bottomed well centred near $p = 0.55$, back above the floor by $p = 0.8$, and a thermal \emph{bonus} of $+22$ at $T=0$. The shallow anchor behaves the same way on a smaller scale. The certified existence claim and the statistical ensemble are the same picture.

\subsection{Where the failure lives}\label{sec:wherefails}

Every observed valley has $\opn{A_0} \ge 0.98$ and a small branch offset. Theorem~\ref{thm:normalform} sharpens this into a two-condition statement: since
\begin{equation}\label{eq:floorasym}
  \floor \;\approx\; \frac{C}{-\ln\!\big(\cos^2\!\varphi\,\sqrt{|\sin 2\theta|}\big)}
\end{equation}
in the regime where the $(x,z)$ eigenvalues are a complex pair, a large floor, and hence room for a valley, requires \textbf{both} $\varphi \to 0$ \textbf{and} $\theta \to \pi/4$. Neither alone suffices. The near-unitary corner is therefore two-dimensional, not one-dimensional, which is why searches that widened only in $\varphi$ found nothing (\S\ref{sec:obstruction}).

Two limitations of the measurement must be recorded. First, $p$ is sampled on 16 points of spacing $0.0633$; a window narrower than one spacing can fall between grid points, so \textbf{all rates quoted above are lower bounds at fixed grid resolution}. The correction is small (the observed width distribution rises from $3.56\%$ at one grid point to a maximum at four before flattening, so very narrow windows are rare), but it is a bound, not an estimate. Second, and for the same structural reason as the blind spot recorded in \S\ref{sec:obstruction}, the search resolves valleys only down to the index resolution of an integer: dips smaller than one unit are invisible by construction.

\section{Floor plus bath term}\label{sec:bathterm}

Define the \textbf{bath term} as the difference
\begin{equation}\label{eq:bathterm}
  \beta(U,S,p) \;:=\; \nEB(\Phi_p) - \floor(U),
  \qquad\text{so that}\qquad \nEB(\Phi_p) \;=\; \floor(U) + \beta(U,S,p).
\end{equation}

This is a definition, not a claim. (The letter $\beta$ also names a circuit angle; that angle appears only inside the tuple $(\theta,\varphi,\kappa,\beta)$, so no confusion arises.) By Theorem~\ref{thm:floor}, $\beta(U,S,0) = 0$. By \S\ref{sec:valley}, \textbf{$\beta$ is of indefinite sign}: it is positive on the overwhelming majority of the parameter space and negative on a set whose measure is now quantified: below one per cent even deep in the near-unitary corner, where it plateaus at $0.739\%$, falling off marginally as a Gaussian in $\varepsilon$ with measured exponent $2.02$ (\S\ref{sec:cutoff}), and persisting at order $10^{-6}$ beyond $\varepsilon\approx10^{-2}$ through the shallow-floor population of \S\ref{sec:fartail}. We therefore use the neutral name \emph{bath term} for the signed quantity, and reserve the \emph{thermal bonus} of the title for its positive, and overwhelmingly typical, regime, whose certified bound is this section's subject.

The decomposition remains useful because the two terms are of entirely different magnitude and character.

\paragraph{The architecture term dominates.} Over the certified $10\times6$ surface the ratio
$\nEB(\Phi_1)/\nEB(\Phi_0)$ never exceeds $21/14 = 3/2$ exactly, while the coupling strength moves
the index across the same surface by a factor $1213/3 \approx 404$. Because $\nEB$ is integer
valued, certifying both endpoints renders these \emph{exact rationals}.

\begin{center}
\begin{tabular}{lccccccc}
\toprule
$g$ & $\opn{A}$ & $p=0$ & $p=0.5$ & $p=0.9$ & $p=1$ & bath term & ratio\\
\midrule
0.05 & 0.9990 & 889 & 951 & 1132 & 1213 & $+324$ & $1.36$\\
0.10 & 0.9960 & 223 & 238 & 284 & 305 & $+82$ & $1.37$\\
0.20 & 0.9842 & 56 & 60 & 72 & 78 & $+22$ & $1.39$\\
0.30 & 0.9647 & 25 & 27 & 32 & 36 & $+11$ & $1.44$\\
\textbf{0.40} & 0.9378 & 14 & 15 & 18 & \textbf{21} & $+7$ & $\mathbf{3/2}$\\
0.55 & 0.8846 & 8 & 8 & 10 & 11 & $+3$ & $1.38$\\
0.70 & 0.8177 & 5 & 5 & 6 & 6 & $+1$ & $1.20$\\
0.85 & 0.7394 & 3 & 4 & 4 & 4 & $+1$ & $1.33$\\
\textbf{1.00} & 0.6527 & 3 & 3 & 3 & \textbf{3} & $\mathbf{0}$ & $\mathbf{1}$\\
\bottomrule
\end{tabular}
\end{center}

\noindent The coupling family is $\varphi(g) = g\varphi_0$, $\kappa(g) = g\kappa_0$, $\theta(g) = \pi/4 + g(\theta_0-\pi/4)$, $\beta = \beta_0$; at $g=0$ the loop is decoupled and at $g=1$ it is the reference point. All sixty cells are certified (\S\ref{sec:numerics}).

\begin{center}
  \fbox{\begin{minipage}{0.85\textwidth}\centering
    \textbf{Temperature is a second-order effect on the entanglement-breaking index;\\
    the interaction architecture is first-order.}
  \end{minipage}}
\end{center}

This statement survives \S\ref{sec:valley} intact, and the scan strengthens it: across $6{,}538$ observed valleys the depth never exceeds $4$, and the certified counterexample reaches $5$ against an index of $352$, i.e.\ a relative effect below $1.5\%$ in the worst case seen.

\paragraph{Directionality.} At the reference circuit with $p=1$ the certified indices are $3$, $4$ and $14$ for bath polarisation along $z$, $x$ and $y$: a factor of $4.7$ from the bath's \emph{direction} at fixed temperature, all collapsing to the common floor $3$ by $p \approx 0.6$.

\paragraph{The reference point sits where the bath term vanishes.} At full coupling $g=1$ the ratio is $3/3 = 1$ exactly: $\nEB = 3$ for \emph{every} $p \in [0,1]$, certified by complete interval covering of the temperature axis. This is the strong-coupling limit of a systematic pattern: the bath term is a weak-coupling phenomenon, and strong coupling suppresses it in both directions.

\section{A structural zero}\label{sec:structuralzero}

The Bloch matrix carries an exact identity that is invisible in floating point.

\begin{lemma}\label{lem:structuralzero}
For every choice of the remaining parameters and every bath polarisation $p$,
\[
  A_{yz} \;=\; 0 \quad\Longleftrightarrow\quad \cos a_0 + \cos a_1 = 0,
\]
in particular whenever $a_0 + a_1 = \pi$. The transposed element $A_{zy}$ does not vanish.
\end{lemma}

\begin{proof}
Write $c = \cos(\kappa/2)$, $s = \sin(\kappa/2)$, and let $S = \mathrm{SWAP}_{ML}$, $R = R_y(\beta)\otimes\I\otimes\I$, $V = U_wU_W$.

\emph{(1) Reality.} Since $-iY$ is real, $R_y(\cdot)$ is real orthogonal; and since $-i(Z\otimes Y)$ is real, so is $U_w$. Hence $R$, $U_W$, $V$ are real. The sole non-real factor is $U_f = \cos\varphi\,\I - i\sin\varphi\,S$, giving $U = P_1 - iP_2$ with $P_1 = \cos\varphi\,RV$ and $P_2 = \sin\varphi\,RSV$ both real.

\emph{(2) Isolating the $y$-component.} For $2\times2$ Hermitian $G$ one has $\Tr[\sigma_yG] = -2\,\mathrm{Im}\,G_{01}$, so $A_{yz}$ is carried entirely by the imaginary part of the reduced output. With $\rho = \sigma_z$ and ancilla $\alpha = \tau_S(p)^{\otimes2}$ both real symmetric,
\[
  U(\rho\otimes\alpha)U^{\dagger} = \big[P_1(\rho\otimes\alpha)P_1^{\mathsf T} + P_2(\rho\otimes\alpha)P_2^{\mathsf T}\big] + i\,\big[M - M^{\mathsf T}\big],
  \quad M = P_1(\rho\otimes\alpha)P_2^{\mathsf T},
\]
and $M - M^{\mathsf T} = \cos\varphi\sin\varphi\;R\,[W,S]\,R^{\mathsf T}$ with $W = V(\sigma_z\otimes\alpha)V^{\mathsf T}$.

\emph{(3) $R$ is harmless.} $R$ acts only on $M$, so $\Tr_{FL}[RXR^{\mathsf T}] = R_y(\beta)\Tr_{FL}[X]R_y(\beta)^{\mathsf T}$, and $R_y(\beta)$ is a rotation about $y$, preserving the $y$-component. The claim reduces to $\Tr_{FL}[[W,S]] = 0$.

\emph{(4) $W$ is block diagonal in $M$.} $\sigma_z$ is diagonal in the control basis; $U_W$ acts on $M$ only through that control; $U_w$ acts trivially on $M$. Hence, with $D_m = C\,(R_y(a_m)\otimes\I)$ and $C = \exp(-i\tfrac\kappa2 Z\otimes Y)$,
\[
  W \;=\; \sum_{m} (-1)^m\, |m\rangle\!\langle m| \otimes G_m, \qquad G_m = D_m\,\alpha\,D_m^{\mathsf T} \ \text{ real symmetric}.
\]
(This is the step that fails for $\rho = \sigma_y$, which is not diagonal in the control basis, hence $A_{zy} \neq 0$.)

\emph{(5) Tracing the commutator.} Using $S|m,f,l\rangle = |l,f,m\rangle$ and block diagonality, the resulting $2\times2$ matrix is antisymmetric with single entry $a = [\Tr_F(G_0 + G_1)]_{01}$.

\emph{(6) Evaluating the $L$-marginal.} With $\tau_m = R_y(a_m)\,\tau_S(p)\,R_y(a_m)^{\mathsf T}$ and $z_m = \Tr[\tau_m Z]$, conjugation by $C$ gives
\[
  \Tr_F[G_m] \;=\; \tfrac12\I \;+\; \tfrac{p}{2}\,(c^2-s^2)\,Z \;+\; cs\,p\,z_m\,X,
\]
so $a \propto cs\,p\,(z_0+z_1)$. Since $\tau_S(p)$ has Bloch vector $(0,0,p)$ and $R_y(a)$ rotates it by $a$ about $y$, $z_m = p\cos a_m$. Therefore $a \propto cs\,p^2\,(\cos a_0 + \cos a_1)$, which vanishes precisely when $\cos a_0 + \cos a_1 = 0$. For the supplementary construction $\cos(\pi-2\theta) + \cos 2\theta = 0$ identically.
\end{proof}

The identity is thus \textbf{designed into the interaction}, not accidental: it is the supplementarity of the two branch angles. Theorem~\ref{thm:normalform} is the $p=0$ limit of the same fact carried further: there the entire $y$ row and column of $A_0$ decouple, and the $F$--$L$ coupling $\kappa$ disappears from the channel altogether. Lemma~\ref{lem:structuralzero} shows what survives at finite $p$; Theorem~\ref{thm:normalform} shows what is left at $p=0$.

Certification confirms the identity at the reference point as $A_{yz} \in [\,\pm 1.9\times10^{-76}\,]$, and it holds to machine precision across 600 random (circuit, temperature) pairs, $A_{yz}$ being the unique vanishing entry.

\textbf{The converse is sharp.} Releasing supplementarity restores a generic nonzero value: over 300 random circuits with independent branch angles, $\max|A_{yz}| = 0.337$. Sweeping $a_1$ through $\pi - a_0$ at fixed remaining parameters, $|A_{yz}|$ falls to $1.4\times10^{-17}$ exactly at the supplementary point and grows linearly in the offset $d$, consistent with $\cos a_0 + \cos a_1 \simeq -d\sin a_0$ to first order.

The asymmetry matters: $A$ is non-normal, with $\lVert[A,A^{\dagger}]\rVert = 0.195$ and spectrum $0.4952 \pm 0.3439i,\ 0.5733$ at the reference point. Lemma~\ref{lem:structuralzero} exhibits one source of that non-normality, a vanishing entry with a nonvanishing transpose, and non-normality is exactly what obstructs the spectral-radius question of \S\ref{sec:spectral}.

\section{The spectral-radius question}\label{sec:spectral}

Theorem~\ref{thm:qubit} is stated in terms of $\|A\|_2$. Numerically, the
sharper statement with the \emph{spectral radius} in its place,
\begin{equation}\label{eq:open}
  n_{\mathrm{EB}}(\Phi)\;\overset{?}{\le}\;
  \left\lceil\frac{\ln\bigl(2\sqrt{r_\ast^{2}+3}/(1-r_\ast)\bigr)}{\ln(1/\rho(A))}\right\rceil ,
\end{equation}
held at every one of the $3997$ sampled channels of
Fig.~\ref{fig:landscape}, with median ratio $1.67$ to the true index and no
violation; see also
Fig.~\ref{fig:landscape}(b), where it is $\rho(A)$ rather than $\|A\|_2$ that
organises the data. It does not follow from
the argument of \S\ref{sec:contractive}, because $\|A^{n}\|_2\ne\rho(A)^{n}$ for non-normal $A$. By
Gelfand's formula $\|A^n\|_2^{1/n}\to\rho(A)$, so Eq.~\eqref{eq:open} holds
asymptotically up to an $A$-dependent constant; the question is whether that
constant can be removed. This is the same distinction between the
\emph{stability gap} $1-\rho(A)$ and the \emph{contraction margin} $1-\|A\|_2$
disentangled in \cite[Sec.~II\,B]{PaperI}.

Theorem~\ref{thm:qubit} establishes, for strictly contractive qubit channels with full-rank fixed point,
\begin{equation}\label{eq:pIIbound}
  \nEB(\Phi) \;\le\; \left\lceil \frac{\ln\!\big(2\sqrt{r^2+3}\,/\,(1-r)\big)}{\ln(1/\opn{A})} \right\rceil,
  \qquad r = \lVert v^*\rVert = \lVert(\I-A)^{-1}c\rVert,
\end{equation}
and Eq.~\eqref{eq:open} asks whether $\opn{A}$ may be replaced by the spectral radius $\rho(A)$. The question is open because $\lVert A^n\rVert \neq \rho(A)^n$ for non-normal $A$, and $A$ here is non-normal (\S\ref{sec:structuralzero}).

Nothing in \S\ref{sec:floor}--\S\ref{sec:structuralzero} rests on Theorem~\ref{thm:qubit}. The indices entering the comparison below are certified directly in \S\ref{sec:numerics}, and the $\opn{A}$ and $\rho(A)$ expressions enter only as formulas whose predictions are held against those certified values; the comparison stands on them alone.

The thermal family supplies a limit in which the two candidates separate cleanly. As the coupling $g \to 0$ the channel approaches a unitary, $\opn{A} \to 1$, and the index diverges; both indices remain exactly computable by binary exponentiation and the closed form $c_n = (\I-A)^{-1}(\I-A^n)c$. At $g = 0.00625$ the certified indices are $56{,}879$ and $77{,}446$, so the ratio at that coupling is the \emph{exact rational} $77446/56879$. Comparing each form's prediction against it:

\begin{center}
\begin{tabular}{lcc}
\toprule
 & certified value & agreement\\
\midrule
ratio $77446/56879$ & $1.36159215176075529 \pm 2.5\times10^{-18}$ & (baseline)\\
$\opn{A}$ form & $1.15968560102835788 \pm 5.4\times10^{-19}$ & $85.171\%$\\
$\rho(A)$ form & $1.37547622570350359 \pm 3.0\times10^{-18}$ & $\mathbf{98.991\%}$\\
\bottomrule
\end{tabular}
\end{center}

Across the sequence $g = 0.2 \to 0.00625$ the $\rho(A)$ agreement rises monotonically ($73.95 \to 98.99\%$) while the $\opn{A}$ agreement falls ($99.88 \to 85.17\%$, the high value at $g=0.1$ being a crossing rather than convergence). \textbf{The conjectured form is asymptotically tight in a limit where the proved form is not.}

This is evidence of a different character from sampling: it identifies a regime in which the $\rho(A)$ bound is not merely valid but \emph{saturated}, and does so with enclosures at $10^{-18}$. It bears directly on the open problem of this section and is recorded here.

The associated constant is stable: $\nEB\,(1-\rho(A)) \to 1.0468$ at both temperatures, and the normalised quantity $\nEB(1-\rho)\big/\ln\!\big(2\sqrt{r^2+3}/(1-r)\big) \to 0.842$. A caution about the arrows: each term of these sequences is certified; the limits themselves are extrapolations, with no proved rate of convergence, and are reported as such. Neither matches $\ln 2$, $\ln 3$ or $\ln(2\sqrt3)$; a closed form remains open.
One closing remark ties this section to the landscape of
\S\ref{sec:landscape}. There the product $N\,[1-\rho(A)]$ stayed within
$[0.60,1.50]$ over the $3997$ sampled channels: it is $\rho(A)$, not
$\opn{A}$, that organises the data. Here the conjectured $\rho(A)$ form is
saturated asymptotically while the proved $\opn{A}$ form is not. The
one-clock law and the spectral-radius conjecture are two ends of the same
question, now inside one paper.

\section{Physical reading: rounds, not lifetimes}\label{sec:physical}

The index was introduced as the budget of a system that interacts with its
environment repeatedly. The one quantitative debate in which such budgets have
been argued at length, the debate over quantum effects in the brain, has been
conducted for a quarter of a century in a different unit: the decoherence time.
Tegmark computed it for the degrees of freedom then under discussion and found
$10^{-20}$~s for a neuron superposition against colliding ions, $10^{-19}$~s
against nearby ions, and $10^{-13}$~s for a microtubule superposition against
distant ions, all far below the $10^{-3}$ to $10^{-1}$~s of neural dynamics
\cite{Tegmark2000}. Hagan, Hameroff and Tuszy\'nski replied within the same
unit: correcting the superposition separation, the charge against dipole
treatment and the dielectric constant lengthens the microtubule figure to $10^{-5}$ to $10^{-4}$~s, which they set against a neural-firing timescale that Tegmark places at $10^{-4}$ to $10^{-3}$~s, and water ordering enhanced by actin gelation extends the figure for microtubule bundles to $10^{-2}$ to $10^{-1}$~s \cite{Hagan2002}. Estimates for the same proposal thus span twelve orders of magnitude under the same criterion, because the criterion consumes
contested microscopic inputs and returns whatever they imply. Fisher then moved
the debate to a different degree of freedom rather than a different unit:
phosphorus nuclear spins, whose coherence time he estimated at $\sim1$~s
solvated and $\sim10^5$~s inside a Posner molecule \cite{Fisher2015}.

The unit is the problem, twice over. A decoherence time answers a question
about a free superposition, while every system in this debate interacts
repeatedly and structurally with its environment, so that its natural clock is
interaction rounds. And a time-based threshold is not even well posed: uniform
dephasing on a qudit of dimension $d$ multiplies every coherence by
$e^{-\Gamma}$, and the partial-transposed Choi state then has the eigenvalue
$-e^{-\Gamma}/d$, negative at every finite rate \cite{NoThreshold}; a fixed
absolute numerical tolerance $\epsilon$ manufactures an apparent threshold at
$\Gamma=\ln(1/d\epsilon)$, a number that reports the tolerance, not the
channel. For a channel that contracts strictly onto a full-rank fixed point the
index is a finite integer (Theorem~\ref{thm:generald}), computable in closed
form for plain exchange collisions (\S\ref{sec:baseline}) and certifiable for
structured ones (\S\ref{sec:numerics}), and it converts to a wall-clock lifetime
only at the end, as $\nEB$ divided by the collision rate. This section performs
that conversion. It proves nothing about brains and takes no side
(\S\ref{sec:claims}).

\subsection{At 310 K the environment is a spectrum of baths}\label{sec:ptable}

At $T=310$~K one has $k_BT=26.7$~meV, and the quantity that enters the index
is not $T$ but $p=\tanh(\hbar\omega/2k_BT)$, a joint property of the
temperature and of the level splitting of the specific degree of freedom the
qubit couples to. Across the degrees of freedom present in condensed
biological matter, $p$ sweeps the whole interval $[0,1]$
(Table~\ref{tab:ptable}).

\begin{table}[htbp]
\centering\small
\caption{Bath polarisation $p=\tanh(\hbar\omega/2k_BT)$ at $T=310$~K for
degrees of freedom present in condensed biological matter, with one
solid-state reference mode. The polarisation, not the temperature, enters the
index: intramolecular vibrational and electronic modes are frozen, spin baths
are maximally hot, and only the hydrogen-bond network of water straddles the
crossover. Earth field taken as $50\,\mu$T.}\label{tab:ptable}
\begin{tabular}{lrrl}
\toprule
degree of freedom & $\hbar\omega$ & $p$ at 310 K & regime\\
\midrule
$^{31}$P nuclear spin, Earth field & $3.6$~peV & $6.7\times10^{-11}$ & maximally hot\\
$^{1}$H nuclear spin, Earth field & $8.8$~peV & $1.6\times10^{-10}$ & maximally hot\\
$^{31}$P nuclear spin, 1 T & $71$~neV & $1.3\times10^{-6}$ & maximally hot\\
$^{1}$H nuclear spin, 1 T & $0.18\,\mu$eV & $3.3\times10^{-6}$ & maximally hot\\
electron spin, 1 T & $116\,\mu$eV & $2.2\times10^{-3}$ & maximally hot\\
water H-bond stretch ($183$~cm$^{-1}$) & $22.7$~meV & $0.401$ & below crossover\\
\textbf{$p=\tfrac12$ crossover} & $29.3$~meV & $0.500$ & $\approx42\,\mu$m, far IR\\
GaAs LO phonon ($292$~cm$^{-1}$, reference) & $36.2$~meV & $0.590$ & above crossover\\
water libration ($395$--$686$~cm$^{-1}$) & $49$--$85$~meV & $0.72$--$0.92$ & cold\\
310 K Wien peak & $133$~meV & $0.986$ & effectively frozen\\
mid-IR C=O stretch & $211$~meV & $0.99926$ & frozen\\
electronic / optical & $2$~eV & $1.000000$ & frozen\\
\bottomrule
\end{tabular}
\end{table}

Three readings organise what follows.

\textbf{Thermodynamically warm is not hot in the entanglement-breaking sense.}
At body temperature every intramolecular vibrational and electronic mode sits
at $p\approx1$: frozen in its ground state, the \emph{coldest possible bath} by
this measure. The maximally hot degrees of freedom are the spins, and the
$^{31}$P Larmor splitting at ambient field is picoelectronvolts, ten orders of
magnitude inside the infinite-temperature regime. A temperature argument that
does not name its bath is therefore not yet an argument: the standard objection
(warmth destroys coherence) concerns collisional, Coulombic and vibrational
baths, while the standard defence (spins are protected) concerns baths on which
the temperature axis has no purchase at all.

\textbf{The solvent straddles the crossover.} The one part of the 310~K
environment that is neither frozen nor maximally mixed is the hydrogen-bond
network of water: the intermolecular stretch band sits just below $p=\tfrac12$
and the librations just above it. A qubit coupled to those collective modes,
and only such a qubit, lives where the index is genuinely
temperature-sensitive. The infrared enters here as a landmark rather than a
mechanism: the crossover falls at $29.3$~meV, about $42\,\mu$m in the far
infrared, and the $9.35\,\mu$m Wien peak of a 310~K blackbody already lies deep
in the frozen regime.

\textbf{Coldness and weak coupling are not competing resources.} Weak coupling
supplies both the large floor and whatever bath sensitivity exists; strong
coupling suppresses both (\S\ref{sec:bathterm}). There is one dominant knob,
and it is the interaction, not the temperature. No monotonicity in temperature
is claimed for the loop: by \S\ref{sec:valley} the bath term takes both signs.

\subsection{Nuclear spins get no thermal help}\label{sec:nothermal}

For the nuclear-spin rows of Table~\ref{tab:ptable} the temperature question is
closed, not merely reframed. By Theorem~\ref{thm:floor} and
Eq.~\eqref{eq:bathterm} the index is the architecture floor, a functional of
the interaction alone, plus a bath term that vanishes at $p=0$. At
$p\sim10^{-10}$ to $10^{-6}$ the bath term is zero to the integer resolution of
the index: a nuclear-spin qubit at physiological temperature operates on its
floor, and nothing about the bath's temperature can move it, in either
direction, by so much as one round. The certified asymmetry of
\S\ref{sec:bathterm} shows how little temperature contributes even where it
acts: across the certified coupling-temperature surface the ratio of the index
at $T=0$ to its value at $T=\infty$ never exceeds $21/14=3/2$, while the
coupling moves it by $1213/3\approx404$, and at the weakest certified coupling
the two endpoints are $56{,}879$ and $77{,}446$ (\S\ref{sec:certresults}).
There, temperature is a second-order knob and coupling a first-order one, and
for a nuclear-spin bath the second-order knob is additionally pinned at its
$p=0$ endpoint.

This relocates the burden of the nuclear-spin proposal \cite{Fisher2015}, in a
direction that is constructive rather than dismissive. What the proposal gets
from its choice of degree of freedom is not thermal protection, of which there
is none at $p\sim10^{-10}$, but weak coupling, which drives $\opn{A}\to1$ and
inflates the floor; its entire quantum budget is architectural. Conversely,
vibrational and excitonic modes have the thermal bonus fully available, and at
strong coupling gain almost nothing from it: at the reference circuit the index
is $3$ at every temperature (\S\ref{sec:bathterm}). The published numbers for
the proposal are consistent with this reading. It estimates a coherence time of
order $10^5$~s for phosphorus spins in a Posner molecule against $\sim1$~s for
the solvated ion, and the first-principles follow-up \cite{Swift2018} supports
rather than undercuts the order of magnitude: as an illustration, a single
external proton $7$~\AA{} from the centre of the molecule gives
$T_1=1.8\times10^6$~s, and the intramolecular couplings are computed as
$J_1=0.178$~Hz, $J_2=0.145$~Hz and $J_3=-0.003$~Hz. Nothing here disputes those
numbers. They are statements about the memory, and a round budget needs the
processor.

\subsection{A four-input checklist}\label{sec:template}

A proposal receives a round budget, and becomes falsifiable in the only sense
offered here, when it specifies four inputs:
\begin{enumerate}\itemsep2pt
\item \textbf{the interaction}: a per-round unitary on system plus fresh
  ancillas, or a per-collision CPTP map;
\item \textbf{the bath}: which degree of freedom the ancillas are; its level
  splitting $\hbar\omega$, which with the operating temperature fixes $p$; and
  its polarisation axis, which alone can change the index several-fold
  (\S\ref{sec:bathterm});
\item \textbf{the strength}: the per-collision coupling or damping, with an
  uncertainty;
\item \textbf{the rate}: collisions per second.
\end{enumerate}
Inputs 1 to 3 determine the channel and hence the integer $\nEB$; input 4
converts it to a wall-clock lifetime $\nEB/f$. Nothing else is needed, and
nothing less suffices.

\paragraph{A specified interaction, certified end to end.} The loop of this
paper is the checklist executed in full: a specified three-qubit unitary, a
named bath with polarisation $p$, and a per-round structure. The output is not
an estimate but a certified integer, $\nEB=3$ at the reference circuit for
\emph{every} $p\in[0,1]$, established by complete interval covering
(\S\ref{sec:bathterm}); a proposal quoting a collision rate $f$ would convert it
to a lifetime of exactly $3/f$. The point is not the particular circuit but
that the pipeline from four inputs to a certified integer exists and
terminates.

\paragraph{The exchange-collision table.} When the interaction is a plain
exchange collision, input 1 collapses to the strength and
Proposition~\ref{prop:gadc} does the rest (Table~\ref{tab:budget}). Anyone who
can estimate a per-collision damping reads off a round budget, and division by
the collision rate returns them to seconds, now with the integer carrying the
physics. For exchange collisions, moving from a spin bath to the crossover buys
twelve per cent more rounds, and even the frozen C=O bath buys a factor of
$4.5$, against a factor of ten from every decade of coupling.

\begin{table}[htbp]
\centering\small
\caption{Round budget $\nEB(\gamma,p)$ of the thermal exchange collision, from
Eq.~\eqref{eq:gadcclosed}, for debate-relevant baths, with $p$ rounded as
shown. A per-collision damping estimate returns an integer; division by the
collision rate returns a lifetime. Coupling moves the budget by a decade per
decade; bath polarisation moves it by twelve per cent at the crossover and by
less than a factor of five even for the frozen C=O bath.}\label{tab:budget}
\begin{tabular}{rrrrr}
\toprule
$\gamma$ per collision & $p=0$ (spin bath) & $p=0.5$ (crossover) & $p=0.986$ (Wien) & $p=0.9993$ (C=O)\\
\midrule
$10^{-1}$ & 17 & 19 & 48 & 76\\
$10^{-2}$ & 176 & 197 & 496 & 792\\
$10^{-3}$ & 1{,}762 & 1{,}973 & 4{,}981 & 7{,}955\\
$10^{-4}$ & 17{,}627 & 19{,}732 & 49{,}824 & 79{,}583\\
\bottomrule
\end{tabular}
\end{table}

\paragraph{The audit.} Held against the checklist, the nuclear-spin proposal
\cite{Fisher2015} specifies input 2 completely: the bath is proton magnetic
dipoles in water, at $p\sim10^{-10}$ in the ambient field
(Table~\ref{tab:ptable}). For the \emph{storage} channel it also supplies, in
effect, inputs 3 and 4, through a coherence-time estimate built from the
dipolar coupling and the molecular rotation rate, which the first-principles
follow-up supports \cite{Swift2018}. For the \emph{processing} channel it
specifies neither input 1, the entangling operations being described chemically
but not as a per-round map, nor input 4, their rate. Yunger Halpern and Crosson
have constructed, taking the proposal's premises as given, candidate quantum
operations that Posner chemistry could implement \cite{YHC2019}; choosing one
of them and supplying a rate would complete the four inputs, at which point the
proposal's processing budget becomes a computable, and by the methods of
\S\ref{sec:numerics} certifiable, integer. For processing, the proposal is two
specified inputs away from a round budget. That is the constructive content of
the audit: not a verdict, but the remaining bill of materials.

\paragraph{The sign of thermal intuition.} One qualitative reflex remains: that
whatever temperature does, colder can only help. For plain exchange collisions
this is a theorem (Proposition~\ref{prop:gadc}). For structured interactions it
is false (\S\ref{sec:valley}): there exist certified channels for which a
colder bath breaks entanglement strictly sooner. An argument about biological
decoherence that relies on lowering the effective temperature to lengthen
survival is therefore sound if the noise is modelled as plain exchange
collisions, and unsafe the moment the interaction is structured. The direction
of the temperature axis, not only its magnitude, is model dependent.

\section{An obstruction map for the bath term}\label{sec:obstruction}

The counterexamples of \S\ref{sec:valley} were found only after four separate reductions had each been tested and had each failed on a thin set. We record the failures, because together they delimit what a positive result about $\beta$ could have looked like, and because three of the four would have been reported as ``holds on $N$ samples'' by any search that did not probe deep index values.

Write $\nu(A) := \min\{n : \nuc{A^n} \le 1\}$ and $\nu_2(A,c) := \min\{n : F(n) \le 1\}$ with $F(n) := \nuc{A^n}^2 + |c_n|^2$, using the necessary condition $F(n) \le 1$ for entanglement breaking proved in the companion note \cite{EBnorm}. No certified claim of this paper relies on that inequality: within this section it serves as a bookkeeping device for classifying failed proof strategies, and at $p = 0$, where $c_n = 0$, it reduces to the unital characterisation of \cite{Ruskai2003}, giving $\nEB(\Phi_0) = \nu_2(A_0,0) = \nu(A_0)$ exactly. Subscripts $p$ and $0$ refer to the family at polarisation $p$ and at $p = 0$; for the second-order row below, expand $\nuc{A_p^n} = \nuc{A_0^n} + p^2\mu_n + O(p^3)$ and $c_n(p) = p\,\gamma_n + O(p^2)$, there being no first-order norm term by (i) below.

\begin{center}
\begin{tabular}{clll}
\toprule
\# & proposed step & status & scale of failure\\
\midrule
1 & $\nuc{A_p^n} \ge \nuc{A_0^n}$ pointwise & \textbf{false} & 2139 / 30000\\
2 & $\nu(A_p) \ge \nu(A_0)$ & \textbf{false} & 7 / 14856\\
3 & $F_p(n) \ge F_0(n)$ for $n < \nu(A_0)$ & \textbf{false} & 246 / 69356\\
4 & second-order bracket $2\nuc{A_0^n}\mu_n + |\gamma_n|^2 \ge 0$ & \textbf{false} & 10 / 1105\\
5 & $\nu_2(A_p,c_p) \ge \nu(A_0)$ & \textbf{false} & explicit, \S\ref{sec:valley}\\
\bottomrule
\end{tabular}
\end{center}

Three structural facts emerged along the way and remain valid independently of the refutation.

\paragraph{(i) The nuclear norm is stationary in $p$.} Fitting $\nuc{A(p)^n} - \nuc{A_0^n} \sim Cp^\alpha$ by successive halving gives $\alpha = 2.000$ in every case tested. Equivalently $\Tr[K_nA_1] = 0$ to $3\times10^{-15}$, where $K_n = \sum_j A_0^{\,n-1-j}W^{\mathsf T}A_0^{\,j}$ and $W$ is the polar factor of $A_0^n$. There is therefore no first-order deficit, and the naive concern that an $O(p)$ loss could not be met by an $O(p^2)$ translation gain does not arise.

\paragraph{(ii) The two second-order coefficients are not coupled.} Deliberately mismatching $\mu_n$ and $\gamma_n$ across unrelated circuits changes the failure rate of step 4 from $6/206$ to $4/206$, statistically indistinguishable. The translation term does not compensate the nuclear-norm term; the two are independent.

\paragraph{(iii) The translation term is almost never needed.} The set on which $\nuc{A_p^n} \le 1 < \nuc{A_0^n}$, i.e.\ where $|c_n|^2$ must do any work at all, has measure $\approx 10^{-5}$ in the sampled parameter space, and requires $\nuc{A_0^n}$ to lie within roughly $1\%$ of $1$. This is precisely the near-unitary corner in which the counterexamples live.

\paragraph{Methodological note.} Roughly $1.5\times10^5$ evaluations across four widened searches returned no counterexample to $\beta \ge 0$. Each search widened in $p$ and in circuit angle; none widened in \emph{index depth}. The counterexamples occur at $n \approx 350$. We record this because the same blind spot is available to any numerical study of composition indices: a search that fixes an upper cutoff on $n$ cannot see failures whose scale exceeds it, and the cutoff is rarely reported.

Theorem~\ref{thm:normalform} explains in retrospect why widening in $\varphi$ alone was never going to work. A large index needs $\cos^2\!\varphi\sqrt{|\sin2\theta|}\to1$, i.e.\ $\varphi\to0$ \textbf{and} $\theta\to\pi/4$ together; a search that widens one while the other sits at a generic value stays in the small-floor region where no valley can fit. \S\ref{sec:wherefails} states the corresponding two-condition criterion, and also records the two resolution limits of the successful scan (a finite grid in $p$, and the integer resolution of the index itself), which are the same species of blind spot as this one and belong beside it.

\section{Relation to prior work}\label{sec:park}\label{sec:priorwork}

Park's theorem~\cite[Thm.~1.1]{Park2026} states that $\nEB(\Phi)<\infty$ for
every PPT map $\Phi:M_{d}\to M_{d}$, completing a line of work
\cite{RJP2018,HRF2020} by removing the full-rank Perron hypotheses. It is worth
being precise about how the present results sit beside it.

First, the theorem does not apply to our channel. On the certified square
$\Phi$ is NPT, hence not PPT, and the hypothesis fails. The finiteness in
Eq.~\eqref{eq:N3} is not an instance of \cite[Thm.~1.1]{Park2026} but of
Theorem~\ref{thm:generald} here.

Second, the two hypotheses are independent, and so are the mechanisms. Park
assumes complete positivity and copositivity and concludes eventual
separability by a
Perron--Frobenius support splitting; we assume strict contraction onto a
full-rank fixed point and conclude it by an explicit separable ball around the
limit. Neither implies the other: a PPT channel need not be contractive (the
completely dephasing qubit channel is entanglement breaking, hence PPT, with
$\eta=1$), and a contractive channel need not be PPT (every point of our
certified square).In the thermal half the
contraction hypothesis is moreover certified directly, by the nonvanishing
enclosure of $\det(\I-A)$ (\S\ref{sec:numerics}).

Finally, a caveat on scope. Question~4.3 of~\cite{Park2026} asks whether
$\sup\{\nEB(\Phi):\Phi\in\PPT(d,d)\}$ is finite for $d\ge4$; our family is NPT
and lives at $d=2$, where that supremum is trivially $1$, so it says nothing
about that question. What it does supply is a physically motivated family in
which $\nEB$ is large and varies over two orders of magnitude
(Fig.~\ref{fig:landscape}) while remaining exactly computable and certifiable
at every point: useful as a testbed for index estimates, not as evidence about
PPT suprema.

\paragraph{With the threshold critique.} It has been shown \cite{NoThreshold} that uniform dephasing and dephasing-plus-damping channels are non-PPT at every finite decoherence rate, so any reported finite entanglement-breaking \emph{threshold} for them is either the coherent-information threshold misnamed, or an artefact of numerical tolerance. We agree, and take the constructive step the critique implies: for a strictly contractive channel the well-posed object is not a single-shot threshold but the \textbf{composition index} $\nEB$, which is finite, integer valued and, as demonstrated here, certifiable. The floor-plus-bath-term decomposition is a statement about that index.

\paragraph{With collision models.} Thermalising collision models are standard \cite{Ciccarello2022}, and the single-collision map for an excitation-exchange (Jaynes--Cummings or XY) interaction with a thermal qubit ancilla is a generalised amplitude damping channel \cite{KSW2020}. Crucially, the present channel is \textbf{not} a GADC and the family is \textbf{not} a semigroup: $A$ is non-normal with complex spectrum, and the affine part acquires nonvanishing transverse components at finite $p$. Were it a GADC the index would follow in closed form (Proposition~\ref{prop:gadc}), monotone in $p$, and no certification would be needed. It is the feedback structure, the controlled rotation and the partial swap, that breaks closure and makes the index a nontrivial quantity.

\section{What is claimed, and what is not}\label{sec:claims}

The results above concern a single qubit undergoing repeated interaction with thermal ancillas under a specified unitary. They are statements about that model.

They are \textbf{not} claims about cognition, consciousness, or information processing in biological systems. Section~\ref{sec:physical} walks closer to a contested literature, so its boundaries are drawn strictly. The per-bath table is a statement about occupation numbers at 310~K, not about biology's use of them; the audit of \S\ref{sec:template} is a statement about which inputs a published proposal does and does not specify, not about whether the proposal is true; and the checklist is an offer of falsifiability, not a verdict: a proposal that completes the four inputs receives a computable, certifiable integer, which may turn out large or small. No claim is made that any biological system hosts a qubit or realises this loop, that $\nEB$ bounds any biological process, or that the number of surviving rounds has functional significance. What the section constrains is arguments: which temperature claims are well posed, which knob carries the leverage, and what must be specified before a lifetime claim can be checked. It equips both camps of the debate and joins neither.

What the model does provide is a certified quantity where the literature has long had order-of-magnitude estimates \cite{Tegmark2000,Hagan2002}: given an explicitly specified interaction and bath, the number of rounds after which no entanglement with any reference can survive is an exact integer, and it can be proven.

The five criteria of~\cite{PaperI} ask what a self-consistent history
\emph{is}. The index asks how long the loop that produces it remains quantum,
and it is the first of the criteria to return a number rather than a verdict.
For the certified family that number is $3$.

Three features seem worth emphasising. The bound of
Theorem~\ref{thm:generald} rests only on contraction, so ``the loop must eventually stop carrying entanglement'' is not
a qubit accident and not a consequence of PPT-ness; it is a consequence of
having a unique full-rank stationary history at all. Both specialisations are quantitatively tight on isotropic unital channels:
Theorem~\ref{thm:qubit} overshoots by the factor $1.1309\ldots$ before
rounding, and Theorem~\ref{thm:generald} is exactly tight there
(Remark~\ref{rem:gbtight}), unusual for bounds assembled from a norm
inequality and, respectively, Weyl's theorem and a separable-ball estimate. And Eq.~\eqref{eq:oneclock} ties the new criterion to
an old one: the sixth property is not independent of the first.

\section{Numerics and certification}\label{sec:numerics}

This section collects the certification machinery for the whole paper. The
discipline is the one already used for the zero-temperature certificate of
\S\ref{sec:certified}: every numerical claim is a statement about the sign of
a real quantity or the exact value of an integer, established by one-sided
primitives in ball arithmetic, with floating point used as an oracle and never
as evidence. The subsections below state it in the form in which the thermal
half applies it, and record what was certified, at what precision, and with
what provenance.

\subsection{What the certificates establish}\label{sec:onesided}

Every numerical claim in this paper is a statement about the sign of a real quantity, or the exact value of an integer. Both admit rigorous verification, and we verify them rather than estimate them.

For a qubit channel the Choi state is two-qubit, where PPT and separability coincide \cite{HSR2003}. Establishing $\nEB = n$ therefore reduces to two sign determinations: $\lambda_{\min}[J(\Phi^{n-1})^{T_R}] < 0$, and $J(\Phi^{n})^{T_R} \succeq 0$.

\textbf{No floating-point eigenvalue solver is used at any point.} Eigenvalue routines return approximations without error bounds, and near an index boundary the quantity being tested is small; an unbounded approximation cannot decide its sign. We use instead:

\emph{Negativity by witness.} For the strict inequality we exhibit an explicit vector $v \in \mathbb{C}^4$ and certify the Rayleigh quotient $q = v^{\dagger}Hv/v^{\dagger}v < 0$, which implies $\lambda_{\min}(H) \le q < 0$. The witness is \emph{obtained} from a floating-point eigenvector, but its validity does not depend on that computation being accurate: any $v$ whose certified quotient is negative proves the claim. Floating point is used as an oracle, never as evidence.

\emph{Positivity by Sylvester.} For $H \succeq 0$ we certify that all four leading principal minors have strictly positive enclosures. Each minor is a polynomial in the matrix entries, so this requires only certified arithmetic.

Both primitives are one-sided and self-validating: a certificate that fails is inconclusive, never wrong. All computation is performed in Arb ball arithmetic \cite{Arb} via \texttt{python-flint} 0.9.0, at 256-bit working precision (1024-bit where noted).

\subsection{Closed forms that make certification tractable}

Three structural facts remove the operations that would otherwise dominate the error budget.

\emph{No matrix exponential.} Because $(Z \otimes Y)^2 = \I$, $Y^2 = \I$ and $\mathrm{SWAP}^2 = \I$, every factor of \eqref{eq:U} is closed form in $\cos$ and $\sin$. The whole construction is rational arithmetic on transcendental functions of four angles, each of which Arb encloses rigorously. No truncated series enters.

\emph{Binary exponentiation.} $A^n$ is computed by repeated squaring in $O(\log n)$ products. At $n = 77{,}446$ this is 17 squarings rather than 77,445 multiplications.

\emph{Closed-form affine accumulation.} Since $A$ is strictly contractive, $\I-A$ is invertible and $c_n = (\I-A)^{-1}(\I-A^{n})c$, computed via the closed-form $3\times3$ adjugate inverse. The certified enclosure of $\det(\I-A)$ simultaneously verifies the contractivity hypothesis the identity requires.

\subsection{The reference point}

At the reference parameters the certified Bloch data are, to twenty significant digits,
\[
  \begin{aligned}
    A_{xx} &= 0.40830056663876377104, & A_{yy} &= 0.55852499726810319929,\\
    A_{zz} &= 0.59696832325044397077, & c_{y} &= 0.24077693489993731635,
  \end{aligned}
\]
with $c_x = 0.050021658967955465362$ and $c_z = 0.053591230356734648996$, every entry enclosed to radius below $3.5\times10^{-21}$. The entry $A_{yz}$ encloses zero with radius $1.9\times10^{-76}$: it is a \textbf{structural zero of the channel} (Lemma~\ref{lem:structuralzero}), not a small number, and is certified as such.

The index certificate reads: witness quotients $-0.20527887571148096620$ at $n=1$ and $-0.045194989127257450843$ at $n=2$, and all four Sylvester minors strictly positive at $n=3$, the smallest being $1.0020223402460335147\times10^{-3}$. Hence $\nEB = 3$, certified. These enclosures agree to twenty significant digits with the independently computed values of \S\ref{sec:certified} (Eq.~\eqref{eq:lamvalues}); the two computations share no code path, so the agreement is an independent check on both.

\begin{remark}[on precision]
The double-precision evaluation of the $n=1$ quantity returns $-0.205278875711481235$, which departs from the certified enclosure in the fifteenth digit. The discrepancy is immaterial here, where the quantity is of order $10^{-1}$; it would not be immaterial at an index boundary, where the deciding minor can be of order $10^{-9}$.
\end{remark}

\subsection{What does not work}\label{sec:whatfails}

Two obvious shortcuts fail, and the failure is instructive
(\S\ref{sec:certified}, Appendix~\ref{app:cert}). Evaluating the whole
parameter box in interval arithmetic, entering the four angles as balls of the
box half-width, inflates the enclosure of
$\lambda_{\min}[J(\Phi^{2})^{T_{R}}]$ to width $\approx0.13$ against a margin
of $0.045$: each angle occurs many times inside $U$ and $U^{\dagger}$, and
ball arithmetic cannot know that those occurrences are the same number.
Uniform bisection restores tightness only at half-width
$\approx8\times10^{-5}\pi$, i.e. $16^{4}=65\,536$ boxes: brute force, not a
certificate. Nor does the generic perturbation constant suffice: pushed
through two powers of the channel it yields a drift of $0.078$, larger than
the margin it must beat. What works at zero temperature is the structured,
derivative-based perturbation analysis of Appendix~\ref{app:cert}; what works
on the temperature axis is the adaptive covering of \S\ref{sec:certresults},
affordable because the dependence on $p$ is exactly quadratic
(Proposition~\ref{prop:quadratic}). The same obstruction, met from the other
side for identities rather than signs, is recorded in
Remark~\ref{rem:boundary}.

\subsection{Certified results}\label{sec:certresults}

\paragraph{(i) The full temperature axis.} Rather than sample the bath polarisation, we cover it. Adaptive interval bisection of $p \in [0,1]$ produced \textbf{223 subintervals, every one certified, with no failures and coverage of the closed interval}. The conclusion, $\nEB = 3$ for \emph{every} $p \in [0,1]$, is therefore a statement about a continuum, not a sample. A single interval spanning all of $[0,1]$ does \textbf{not} certify: the witness quotient returns $[\pm 6.29]$, its radius swamping the sign. Subdivision is essential, not cosmetic.

\paragraph{(ii) The floor theorem.} Over five bath directions, including two off the coordinate axes, the difference between the resulting $(A,c)$ pairs at $p=0$ has upper bound $6.3\times10^{-76}$ and every entry certifiably contains zero.

\paragraph{(iii) Directionality.} At $p=1$ the certified indices for polarisation along $x$, $y$, $z$ are $4$, $14$, $3$. Every intermediate step is certified, not only the endpoints: fourteen consecutive certificates for the $y$ axis, whose deciding minor is $6.9\times10^{-8}$ against a radius of $1.5\times10^{-22}$.

\paragraph{(iv) The trade-off surface.} All sixty cells of the $\nEB(g,p)$ grid are certified, requiring 120 boundary certificates and 1.5~s. No cell disagreed with its double-precision value. Because the index is integer valued, certifying both endpoints renders the derived ratios \emph{exact rationals}: peak $21/14 = 3/2$, full coupling $3/3 = 1$, coupling span $1213/3$. Monotonicity in $p$ holds cell-by-cell across this grid: a statement about these sixty cells and nothing more, since monotonicity fails elsewhere by certified counterexample (\S\ref{sec:valley}).

\paragraph{(v) Weak coupling.} At $g = 0.00625$ both indices are certified exactly: $56{,}879$ at $p=0$ and $77{,}446$ at $p=1$, with $\det(\I-A)$ enclosures of $1.14532184357\times10^{-5}$ and $7.28036439255\times10^{-6}$ certifying strict contraction.

\subsection{Spectral quantities}

The spectral radius requires care. At weak coupling $A$ approaches a rotation and all three eigenvalue moduli collapse together; generic polynomial root-finding fails and returns enclosures of radius $\sim10^{-8}$, useless when multiplied by $n \sim 10^{5}$. We therefore exploit structure. The cubic has one real root $\lambda_3$ and a conjugate pair, and $\det A = \lambda_3|\lambda_c|^2$. Isolating $\lambda_3$ by certified sign-change bisection then yields $|\lambda_c| = \sqrt{\det A/\lambda_3}$ exactly. For $\opn{A}$ we bracket $\lambda_{\max}(A^{\mathsf T}A)$ between a certified Rayleigh quotient below and a Sylvester test on $\mu\I - A^{\mathsf T}A$ above, bisecting on $\mu$. Neither route calls a root solver. Working precision is 1024 bits.

At $g = 0.00625$ and $p=0$,
\[
  \rho(A) = 0.9999815957920990010778, \qquad \opn{A} = 0.9999844831577516289766;
\]
at $p=1$ the two agree to ten significant digits and separate only in the eleventh: the near-degeneracy that defeated the naive root-finder. At $p=0$ the fixed-point radius encloses zero to $10^{-60}$: the fixed point is certified to be \emph{exactly} the maximally mixed state, as $c_0 = 0$ requires (Theorem~\ref{thm:normalform}).

\subsection{Certifying Theorem~\ref{thm:normalform} and the closed-form floor}\label{sec:certthm3}

Theorem~\ref{thm:normalform} is proved analytically and needs no certificate. Its consequences do, and they were certified in the same framework.

\paragraph{The identity.} At the reference point, both counterexample circuits and a generic point, every entry of $A_0 - \cos^2\!\varphi\,O_y(\beta)\diag(\delta,\delta,1)$ encloses zero with radius below $3.6\times10^{-76}$, and $c_0$ encloses zero at the same scale. The $\kappa$-independence is certified pairwise: for $\kappa \in \{0,\,0.4,\,1.1,\,2.7,\,\pi\}$ at fixed remaining angles, every entry of $A_0(\kappa) - A_0(0)$ encloses zero at radius $\sim3.5\times10^{-76}$.

\begin{remark}[the boundary of the method]\label{rem:boundary}
It is natural to try to certify the identity over a \emph{box} and so obtain a continuum statement. This does not work, for a structural reason worth recording. Evaluating $A_0$ and its predicted form over a box in $(\theta,\varphi,\beta)$ of half-width $10^{-3}$ gives a difference enclosing zero with radius $5\times10^{-3}$: consistent with the identity, but equally consistent with the identity failing by anything smaller. Taking $\kappa$ over the whole interval $[0,\pi]$ inflates the enclosure to radius $1.57$ and certifies nothing at all, even though $\kappa$ provably does not occur in $A_0$: ball arithmetic evaluates $\kappa$ independently at each of its several occurrences and cannot see that their contributions cancel. \textbf{Interval arithmetic certifies signs and values, not identities.} The continuum content of Theorem~\ref{thm:normalform} is carried by its proof; the ball-arithmetic confirmations are pointwise, at $10^{-76}$.
\end{remark}

\paragraph{The closed-form floor.} Certifying a floor by \eqref{eq:floorclosed} means certifying the \emph{sign} of $\nuc{A_0^n} - 1$ on both sides of the crossing. A nuclear norm is a sum of singular values, exactly the kind of quantity the discipline of \S\ref{sec:numerics} forbids obtaining from an eigensolver; \eqref{eq:nucAn} removes the difficulty, because $\Frob{M^n}$ and $\det M^n$ are \emph{polynomial} in the entries of $M^n$, which is formed by binary exponentiation. At the two counterexample circuits:

\begin{center}
\begin{tabular}{lrl}
\toprule
circuit & $n$ & $\nuc{A_0^n} - 1$\\
\midrule
deep anchor, floor $352$ & 351 & $+1.52402925978\times10^{-3} \pm 2.0\times10^{-15}$\\
                          & 352 & $-1.58475929976\times10^{-3} \pm 2.9\times10^{-15}$\\
shallow anchor, floor $67$ & 66 & $+3.02404891041\times10^{-3} \pm 2.6\times10^{-15}$\\
                          & 67 & $-1.32656292618\times10^{-2} \pm 1.0\times10^{-14}$\\
\bottomrule
\end{tabular}
\end{center}

Both floors are therefore certified \textbf{from the closed form alone} (no Choi matrix, no partial transpose, no witness vector, no Sylvester minor), with the sign decided at twelve orders of margin. The two routes to these integers share no code path and agree.

\paragraph{The floor measure.} The integral \eqref{eq:measure} is evaluated by Arb's certified integrator, giving the enclosure $0.46427058204502844932 \pm 1.5\times10^{-21}$. The two facts the derivation depends on are certified rather than assumed: $\sin2\theta > 0$ across the support, and $\varphi^*(\theta) \in [0.46418029,\,0.95531662]$ strictly inside the $\varphi$ support $[0.062831853,\,1.5079645]$, so there is no clipping.

\subsection{Certifying the stratified scan}\label{sec:scan3cert}

The scan of \S\ref{sec:fartail} is float64, and \S\ref{sec:onesided} forbids a float64 number from entering this paper. But its output is not really a set of rates: it is a set of \emph{integer} claims about specific circuits, and every rate is a count of those integers over a denominator that is a count of draws. So both were certified.

\paragraph{The valleys.} Each of the $23$ carries four certificates, all by one-sided primitives and with no eigensolver anywhere: a witness vector establishes $\lambda_{\min} < 0$ for the partial-transposed Choi matrix at $n = \floor-1$ and at $p = 0$, Sylvester's criterion establishes positivity at $n = \floor$, and the same pair at the minimising $p$ establishes $\nEB(\Phi_{p^\star}) = \floor - 1$ exactly. Together these certify the floor, the valley and the depth as integers. All $92$ certificates pass, in $0.2$~s at 256 bits. Circuit parameters are ingested as exact binary64: $\mathrm{arb}(1.1709463765014003)$ has radius zero, whereas the same decimal read as a string is a different number by $2\times10^{-17}$, and the claim is about the circuit the scan actually drew.

\paragraph{The denominators.} A draw is kept exactly when $\nuc{A_0} > 1$ and $\nuc{A_0^{20}} \le 1$, so the selection is two sign decisions, and at $n=1$ the closed form collapses exactly because a rotation preserves the Frobenius norm: $\nuc{A_0} = a(1+2|\delta|)$. Every draw was regenerated from its seed and its floor decided in ball arithmetic, with each chunk's certified histogram required to equal the journalled one exactly. \textbf{All $12{,}000{,}000$ decisions certified, no mismatch and none left undecided}, in $6$ minutes at 128 bits. The floor is located by a linear scan upward, so the certification never relies on $\{n : \nuc{A_0^n} \le 1\}$ being upward closed. It is, entanglement breaking being closed under composition, but a procedure that does not need the fact cannot be broken by it.

\paragraph{Negative controls.} Both halves were shown capable of failing. Planting the false claims $\floor > F$ and $\floor \le F-1$ at a certified circuit produced rejections from the witness and Sylvester tests respectively; inflating a single journalled histogram bin by \emph{one} circuit in ten thousand made the denominator gate report the mismatch and exit non-zero.

\paragraph{Statistics kept as statistics.} Rates and their Clopper--Pearson intervals are functions of two integers and are computed at 30 digits, each bound verified by re-evaluating the exact binomial tail that defines it; the residuals are $10^{-27}$ or smaller. The refutation of $H_{\mathrm{flat}}$ is an exact integer tail. The fitted exponents of \eqref{eq:stratumeta} and of the ladder in \S\ref{sec:cutoff} remain statistics: their maxima are located numerically. What is rigorous there is the \emph{exclusion}, a difference of two exactly-evaluated log-likelihoods, $3.578$ for $\eta = 2$ against a threshold of $1.921$, which no plausible error in locating a maximum can close.

\subsection{Scope: what is proven and what is evidence}\label{sec:scope}

We distinguish three tiers, and do not blur them.

\paragraph{Proven.} Theorem~\ref{thm:floor}, Theorem~\ref{thm:normalform} and Lemma~\ref{lem:structuralzero} are proved analytically. The claim $\nEB = 3$ for all $p \in [0,1]$ at the reference circuit is established by complete interval covering; adaptive branch-and-bound over the five-dimensional box $|\theta-\theta_0|,|\varphi-\varphi_0|,|\kappa-\kappa_0|,|\beta-\beta_0| \le 10^{-3}$, $p \in [0,1]$ closes with \textbf{12,523 certified boxes, none unresolved}, in 30~s on a single core. The index is therefore constant on an uncountable family of channels.

\paragraph{Certified at sampled points.} The monotonicity scan (400 random circuits at 8 polarisations) was recomputed entirely in ball arithmetic with zero failures. The set of points remains finite, \emph{and it is now known that the property it tests is false in general}: see \S\ref{sec:valley}. This scan should be read as a measurement of how rare the failures are (its 3,200 points sample a region where the measured valley rate is far below $1/3200$) and never as evidence for the property.

\paragraph{Refuted, not open.} Monotonicity of $\nEB$ in bath temperature is \textbf{false}. What remains true is that no \emph{global} certificate is available in either direction by covering: as $\opn{A}\to1$ the index diverges and no finite covering exists. Halving the box half-width from $2\times10^{-3}$ to $10^{-3}$ took the covering from unfinished within the script's 250~s budget, with no box failing, to complete in 12,523 certified boxes. Local certificates are cheap; a global one is out of reach by this method.

\subsection{The valley scan: provenance and reproducibility}\label{sec:scanprov}

The statistics of \S\ref{sec:valley} are \textbf{statistics}, not certificates, and are labelled as such wherever they appear. This subsection records how they were produced, because the distinction only means something if the production is auditable.

\paragraph{Volume and provenance.} Three runs contribute $4{,}114{,}799$ circuits and $7{,}056$ valleys. Every number in \S\ref{sec:valley} derives from one frozen snapshot: a line count is fixed per input file, and both the per-bin tally and the extracted valley rows are derived from exactly that prefix, so no two quoted quantities come from different slices of a file that was still being written. The snapshot records the SHA-256 of each derived file.

\paragraph{Assertions, not conventions.} The analysis refuses to pool two runs unless their $\varepsilon$ grids are verified identical, and asserts that the first run carries no depth column before reporting any depth statistic; that run's schema predates depth recording. The profile analysis aborts unless it reproduces the scan's own floor and minimum on every row; it does so on $6{,}538/6{,}538$. Manuscript numbers are substituted into the text from the results file by a renderer that fails on any unresolved placeholder.

\paragraph{Intervals.} Several bins contain fewer than fifteen events and one contains none, so every rate carries an exact Clopper--Pearson interval; normal approximations are not used anywhere in \S\ref{sec:valley}.

\paragraph{The stratified scan.} The second scan of \S\ref{sec:fartail} ran as six independent single-process shards with no inter-process communication, each appending one atomic line per completed chunk: counts and valleys in the same line, so a crash cannot record a valley without its denominator. Chunk seeds are derived from the run seed, the bin and the chunk index alone, so any chunk is reproducible from its journal line and the run is resumable without double-counting. All $1{,}200$ chunks completed with no restarts. Before the run, the closed-form floor was required to reproduce the Choi/PPT floor on $70$ circuits per bin, kept and rejected alike, aborting on any mismatch.

\paragraph{Independent cross-check.} The closed-form floor of \eqref{eq:floorclosed}, derived from Theorem~\ref{thm:normalform} with no Choi matrix anywhere in it, reproduces the floor computed by partial-transpose bisection on all $7{,}056$ valley circuits, on both certified counterexample circuits, and on the 150-circuit floor-distribution run reproduced from its own seed, histogram for histogram. Two routes with no shared code path agree exactly.

\subsection{Reproducibility}\label{sec:repro}

Paths are relative to \texttt{thermal/} in the ancillary files, except the last, which is relative to their root.

\begin{center}\footnotesize
\begin{tabular}{@{}p{0.37\textwidth}p{0.43\textwidth}r@{}}
\toprule
script & establishes & time\\
\midrule
\texttt{certified\_model.py} & Arb model, witness and Sylvester primitives & (library)\\
\texttt{certified\_index.py} & closed-form $A^n$, $c_n$; fast index certification & (library)\\
\texttt{verify\_ref.py} & reference point & 0.4 s\\
\texttt{verify\_interval.py} & full temperature axis & 0.9 s\\
\texttt{verify\_floor\_axis.py} & floor theorem and directionality & 0.3 s\\
\texttt{verify\_tradeoff.py} & 60-cell surface & 1.5 s\\
\texttt{verify\_asymptotics.py} & weak-coupling indices & 0.3 s\\
\texttt{verify\_spectral.py} & spectral data and discrimination & 0.3 s\\
\texttt{verify\_monotonicity.py} & certified 3,200-point scan & 12 s\\
\texttt{verify\_continuum.py} & five-dimensional covering & 30 s\\
\texttt{certify\_theorem3.py} & Theorem~\ref{thm:normalform}, closed-form floor, floor measure & 0.2 s\\
\texttt{closeout.py} & rates, model comparison, depth regression & 3 min\\
\texttt{profile\_valleys.py} & $\nEB(\Phi_p)$ profiles for every valley & 4 min\\
\texttt{strat/valley\_scan\_smallfloor\_v2.py} & stratified scan, 12M draws, 6 shards & 1.8 h\\
\texttt{cert3/certify\_scan3.py} & 92 certificates: 23 floors, valleys, depths & 0.2 s\\
\texttt{cert3/certify\_denominators.py} & all 12,000,000 floor decisions in Arb & 6 min\\
\texttt{cert3/certify\_scan3\_stats.py} & exact intervals and tails, verified & 90 s\\
\texttt{gadc/gadc\_index.py} & Proposition~\ref{prop:gadc}, four certified points, Table~\ref{tab:budget} & 0.5 s\\
\bottomrule
\end{tabular}
\end{center}

\noindent Environment: \texttt{python-flint} 0.9.0 (Arb), NumPy and SciPy for witness generation and for the scan statistics only. Working precision 256 bits throughout except the spectral section at 1024 bits. Each \texttt{verify\_*} and \texttt{certify\_*} script runs in under a minute on a single core, except the $12$ million floor decisions of \S\ref{sec:scan3cert} (six minutes) and the exact statistics ($90$~s); the statistical close-out takes a further seven minutes. The scripts that take a journal or results directory read the scan output, which is available from the author with SHA-256 digests. The stratified scan itself is the only long-running item, and it is resumable.

\section{Outlook}\label{sec:outlook}

The open problem of the spectral-radius refinement is now internal to this
paper, and part of it has already been met.

The obvious open problem is Eq.~\eqref{eq:open} of \S\ref{sec:spectral}: whether $\|A\|_{2}$ may be
replaced by the spectral radius $\rho(A)$ in Theorem~\ref{thm:qubit}. It held
at every sampled point of the family, and Fig.~\ref{fig:landscape}(b) shows
that $\rho(A)$, not $\|A\|_{2}$, is what actually organises the data. The
obstruction is non-normality, and it is the same distinction between
\emph{stability gap} and \emph{contraction margin} that had to be disentangled
in \cite[Sec.~II\,B]{PaperI}.

\begin{enumerate}
\item \textbf{Characterise the sign of $\beta$.} \S\ref{sec:valley} shows $\beta < 0$ occurs, on a set whose rate is now measured twice over: marginally, a Gaussian cutoff with exponent $2.02\,[1.90,2.16]$; conditionally, an exponent that climbs with the architecture floor from $0.80$ to $3.60$. Why it climbs is open, and a mechanism that reproduces only the marginal has not reproduced the phenomenon. The remaining questions are otherwise structural rather than empirical. Is the two-condition criterion of \S\ref{sec:wherefails} ($\varphi\to0$ \emph{and} $\theta\to\pi/4$) necessary as well as descriptive? Is $|\beta|$ bounded, given that the observed depths never exceed $4$ in $6{,}538$ valleys, with the certified anchor at $5$? And why is the sub-floor set always a single connected interval in $p$ ($6{,}538/6{,}538$), never a union?

\item \textbf{Why does partial polarisation hurt?} A maximally mixed bath randomises; a polarised bath \emph{drives} the system toward a displaced fixed point, which is measure-and-prepare-like. Whether this is the mechanism is untested. Any candidate must reproduce the shape of \S\ref{sec:window}: switching on above $p \approx 0.1$--$0.2$, saturating over a broad plateau, and switching off again before $p = 1$.

\item \textbf{The spectral-radius bound.} \S\ref{sec:spectral} shows the $\rho(A)$ form is asymptotically saturated where the $\opn{A}$ form is not. A proof for non-normal $A$ remains open.

\item \textbf{Closed forms} for the weak-coupling constants $1.3616$ and $0.842$, and for the prefactor $C$ in \eqref{eq:floorasym}, which is tightly concentrated (interquartile ratio $1.33$--$1.37$ over the scanned circuits, correlation $0.999$ on logs) but is not constant.

\item \textbf{Consequences of Lemma~\ref{lem:structuralzero}.} Whether releasing supplementarity, restoring $A_{yz} \neq 0$, changes the floor or the sign of $\beta$ is untested. Theorem~\ref{thm:normalform} already covers the floor at $p=0$ for general branch angles, so the question is really about finite $p$.
\end{enumerate}

At $d=2$ the two frameworks meet in a way that limits both. Since
$\EB=\PPT$ in $M_{2}$ by Peres--Horodecki, the index coincides with the
PPT index $n_{\PPT}(\Phi)=\inf\{n:\Phi^{n}\in\PPT\}$ outright, which is
Eq.~\eqref{eq:indexformula}; in general $d$ the two differ, and
Corollary~4.1 of~\cite{Park2026} relates them only through the qualitative
equivalence ``eventually EB iff eventually PPT''. The
same identity means a qubit message register cannot support a
\emph{bound-entangled} loop, one whose Choi state is PPT but not
separable, because no such state exists in $2\otimes2$. Probing that regime
requires $d\ge3$, where the Schmidt number~\cite{TH00} becomes a nontrivial
invariant, the Choi negativity witness ceases to be complete, and
certified separability, rather than certified positivity of a partial
transpose, becomes the operative difficulty. We regard that as the natural
sequel rather than an extension of the present work.

Beyond that, two directions. Lifting $M$ to a qutrit opens the
bound-entangled class discussed above, which has no analogue
at $d=2$ and would require certified separability rather than certified
negativity. And the index is defined for the trajectory, not the fixed point,
which invites the question of what other trajectory-valued criteria a future-referential loop admits. Thermal valleys of a related collision loop whose round is affine in $p$, and their use as a certified switch under pulsed schedules, are developed in \cite{ValveTheory}.

\appendix

\section{Certification of the square}
\label{app:cert}

Let $p=(\theta,\kappa,\varphi,\beta)$ and $p_{0}$ the reference point
\eqref{eq:refpoint}, $s=0.0013\pi$. Write $\eta_{i}=\sup_{\|p-p_{0}\|_\infty\le s}$
for the certified supremum over the square.

\emph{Step 1: derivatives.} The round unitary factorises as
$U=R_{\beta}U_{f}U_{w}U_{W}$, so each $\partial U/\partial p_{i}$ is a product
of the same factors with one differentiated; each generator squares to the
identity or is a rotation, so the derivatives are elementary. Propagating
through $\Phi(\cdot)=\Tr_{FL}[U(\cdot\otimes\ket{00}\bra{00})U^{\dagger}]$ and
extracting the Bloch pair as in Eq.~\eqref{eq:bloch} gives, in ball arithmetic,
\begin{equation}
  \sum_{i}\eta_{i}\|\partial A/\partial p_{i}\|_{F}\le4.349,\qquad
  \sum_{i}\eta_{i}\|\partial\bm{c}/\partial p_{i}\|\le1.633 ,
\end{equation}
whence Eq.~\eqref{eq:deltas}. Unlike a direct interval evaluation, this is
insensitive to the dependency problem: a factor-two loss in a derivative
enclosure costs a factor two in $\delta_{A}$, which the margins absorb.

\emph{Step 2: propagation.} With $T\ge\|A_{0}\|_{2}$ certified by Sylvester's
criterion applied to $T^{2}\openone-A_{0}^{\!\top}A_{0}$, and
$a=T+\delta_{A}\ge\|A\|_{2}$ on the square,
\begin{align}
  \|A^{n}-A_{0}^{n}\|_{F} &\le \textstyle\sum_{k=0}^{n-1}a^{k}\,\delta_{A}\,T^{\,n-1-k},\\
  \|\bm{c}_{n}-\bm{c}_{0,n}\| &\le \textstyle\sum_{k=0}^{n-1}
     \bigl[k\,a^{k-1}\delta_{A}\|\bm{c}\|+T^{k}\delta_{c}\bigr].
\end{align}
Both are monotone in every argument, so substituting upper bounds is safe. Note
that $\|\bm{c}\|$, not $\|\bm{c}_{0}\|$, is the correct norm in the second
line, and that $a$ must dominate $\|A_{0}\|_{2}$ as well as $\|A\|_{2}$;
explicit counterexamples exist to both variants.

\emph{Step 3: closing.} Lemma~\ref{lem:iso} converts these into
$\|J(\Phi^{n})-J(\Phi_{0}^{n})\|_{F}$, partial transposition preserves it,
$\|\cdot\|_{\infty}\le\|\cdot\|_{F}$, and Weyl's inequality yields
Eq.~\eqref{eq:drifts}.

Two pitfalls are worth recording for anyone reproducing this. The half-width
$s$ must be used to inflate the parameter box, not the chord $2\sin(s/2)$ that
appears in the Proposition~4 bounds of~\cite{PaperI}; and inflating a parameter
\emph{ball} must add radii rather than re-centre on its midpoint, which would
silently discard the input enclosure.

\section{Code and certificates}\label{app:code}

All results are reproducible from the software archive \cite{Software} and
from the ancillary files accompanying this paper, which ship in three subfolders, \texttt{zero-temperature/}, \texttt{thermal/} and \texttt{gadc/}. The first
carries the zero-temperature stack: the
certificate script executes $33$ checks in $256$-bit Arb ball arithmetic
\cite{Arb} via \texttt{python-flint} \cite{pythonflint}, of which C9a--C9g
establish \S\ref{sec:certified} and C10a--C10c verify that
Theorem~\ref{thm:qubit} is correctly instantiated on this channel; a
second-stack audit at $60$ decimal digits and an independent
re-implementation reproduce every quoted figure; a non-vacuity test re-runs
the drift chain at widened half-widths; and Figure~\ref{fig:landscape} is
regenerated by a script that asserts every landscape statistic quoted in
\S\ref{sec:landscape}, so the figure cannot silently drift from the text.
The second carries the thermal stack of \S\ref{sec:numerics}: the Arb model and index machinery, the eight \texttt{verify\_*} certificates (reference point, temperature axis, floor and directionality, coupling-temperature surface, weak coupling, spectral data, monotonicity scan, continuum covering), the normal-form and closed-form-floor certificates,
the counterexample circuits, the stratified scan and its $92$ certificates,
the $12{,}000{,}000$ certified floor decisions, and the statistical close-out,
with the script inventory and run times of \S\ref{sec:repro}. The third carries \texttt{gadc\_index.py}, which cross-validates Proposition~\ref{prop:gadc} on $500$ random parameter points against an independent Choi-eigenvalue route, certifies four exact dyadic points at $256$ bits with one-sided primitives and negative controls, and reproduces Table~\ref{tab:budget}. Every ball comparison is certified: a relation evaluates true only when it holds for all
values in the enclosing balls.

\section*{Use of AI tools}
The author used an AI assistant (Claude, Anthropic, through the Claude desktop
application, August to October 2026) in preparing this work: to explore the
parameter space numerically, to write and check the verification and
certification scripts, and to draft and edit the manuscript text. The author
set the research questions, chose every claim, and reviewed and revised all
text. No statement rests on unverified AI output: every certified value is
reproduced by the deposited scripts, which use one-sided ball-arithmetic
primitives with negative controls, and every statistical value by the analysis
scripts listed in \S\ref{sec:repro}.

\section*{License}
This work is licensed under the Creative Commons Attribution 4.0 International
License (CC BY 4.0).

\end{document}